\documentclass{siamart251216}
\usepackage{amsmath,amssymb,amsfonts}
\usepackage{braket}      
\usepackage{orcidlink}
\usepackage{tikz}
\usetikzlibrary{calc}
\usepackage{bbm}
\usepackage{graphicx}    
\usepackage{comment}
\usepackage{subcaption}
\usepackage{mathtools}
\usepackage{MnSymbol}
\usepackage{xcolor}
\usepackage{environ}
\usepackage{nicematrix}

\usepackage{pgfplots}
\usepgfplotslibrary{groupplots}
\pgfplotsset{compat=1.18}

\newcommand{\R}{\mathbb{R}}
\newcommand{\norm}[1]{\lVert #1 \rVert}
\newcommand{\polylog}{\operatorname{polylog}}
\newcommand{\Lt}{\widetilde{L}}
\newcommand{\Rt}{\widetilde{R}}
\newtheorem{rem}{Remark}[theorem]
\usepackage{qcircuit}  
\newcommand{\myqctmp}[2][0.25]{\Qcircuit @C=#2em @R=#1em @!R}
\NewEnviron{myqcircuit}[1][0.25]{\vcenter{\myqctmp[#1]{0.5} {\BODY}}}

\usepackage[draft,commandnameprefix=ifneeded]{changes}  
\definechangesauthor[name={GA}, color=teal]{ga}
\definechangesauthor[name={GDC}, color=blue]{gdc}

\usepackage{xcolor}

\usepackage{adjustbox}

\definecolor{col1}{RGB}{235,25,40}
\definecolor{col2}{RGB}{0,120,240}
\definecolor{col3}{RGB}{0,205,90}
\definecolor{col4}{RGB}{160,30,235}
\definecolor{gg}{RGB}{130,130,130}   
\definecolor{zz}{RGB}{215,215,215}   

\newcommand{\opBD}[1]{%
  \begin{tikzpicture}[baseline=(bb.center)]
    \def\cw{0.14}\def\ch{0.058}
    \fill[zz] (0,0) rectangle (4*\cw,-64*\ch);
    \foreach \kk in {0,1,2,3}{%
      \foreach \cc/\ss [count=\m from 0] in {#1}{%
        \pgfmathsetmacro{\yy}{-(16*\kk+\m)*\ch}\pgfmathsetmacro{\xx}{\kk*\cw}%
        \fill[\cc!\ss] (\xx,\yy) rectangle (\xx+\cw,\yy-\ch);}%
    }%
    \foreach \vk in {1,2,3}{\draw[black,thin](\vk*\cw,0)--(\vk*\cw,-64*\ch);}%
    \foreach \bk in {1,2,3}{\pgfmathsetmacro{\yb}{-16*\bk*\ch}\draw[white,semithick](0,\yb)--(4*\cw,\yb);}%
    \draw[black,very thick](0,0) rectangle (4*\cw,-64*\ch);
    \node (bb) at (2*\cw,-32*\ch){};
  \end{tikzpicture}%
}

\newcommand{\opBDfour}[4]{%
  \begin{tikzpicture}[baseline=(bb.center)]
    \def\cw{0.14}\def\ch{0.058}
    \fill[zz] (0,0) rectangle (4*\cw,-64*\ch);
    \foreach \cc/\ss [count=\m from 0] in {#1}{\pgfmathsetmacro{\yy}{-\m*\ch}\fill[\cc!\ss](0,\yy) rectangle (\cw,\yy-\ch);}%
    \foreach \cc/\ss [count=\m from 0] in {#2}{\pgfmathsetmacro{\yy}{-(16+\m)*\ch}\fill[\cc!\ss](\cw,\yy) rectangle (2*\cw,\yy-\ch);}%
    \foreach \cc/\ss [count=\m from 0] in {#3}{\pgfmathsetmacro{\yy}{-(32+\m)*\ch}\fill[\cc!\ss](2*\cw,\yy) rectangle (3*\cw,\yy-\ch);}%
    \foreach \cc/\ss [count=\m from 0] in {#4}{\pgfmathsetmacro{\yy}{-(48+\m)*\ch}\fill[\cc!\ss](3*\cw,\yy) rectangle (4*\cw,\yy-\ch);}%
    \foreach \vk in {1,2,3}{\draw[black,thin](\vk*\cw,0)--(\vk*\cw,-64*\ch);}%
    \foreach \bk in {1,2,3}{\pgfmathsetmacro{\yb}{-16*\bk*\ch}\draw[white,semithick](0,\yb)--(4*\cw,\yb);}%
    \draw[black,very thick](0,0) rectangle (4*\cw,-64*\ch);
    \node (bb) at (2*\cw,-32*\ch){};
  \end{tikzpicture}%
}

\newcommand{\Smatrix}{%
  \begin{tikzpicture}[baseline=(bb.center)]
    \def\c{0.5}
    \foreach \cc/\ss [count=\k from 0] in {%
      col1/100,col1/83,col1/67,col1/50,
      col1/83,col2/83,col2/67,col2/50,
      col1/67,col2/67,col3/67,col3/50,
      col1/50,col2/50,col3/50,col4/50}{
      \pgfmathtruncatemacro{\row}{floor(\k/4)}\pgfmathtruncatemacro{\colu}{mod(\k,4)}
      \fill[\cc!\ss] (\colu*\c,-\row*\c) rectangle (\colu*\c+\c,-\row*\c-\c);}
    \draw[black,thin] (0,0) grid[step=\c] (4*\c,-4*\c);
    \draw[black,very thick] (0,0) rectangle (4*\c,-4*\c);
    \node (bb) at (2*\c,-2*\c){};
  \end{tikzpicture}%
}

\definecolor{col1}{RGB}{220,60,60}
\definecolor{col2}{RGB}{60,150,220}
\definecolor{col3}{RGB}{60,200,100}
\definecolor{col4}{RGB}{150,60,220}
\definecolor{graybox}{RGB}{180,180,180}

\headers{Block encoding of semiseparable matrices}{G. Antonioli, G. M. Del~Corso, R. Van~Beeumen}

\title{Block encoding of semiseparable matrices via their generators\thanks{Submitted to the editors DATE.}
\funding{G.Antonioli and G. Del Corso  carried out
this study  within the National Centre on HPC, Big Data and Quantum Computing - SPOKE 10 (Quantum Computing) and received funding from the European Union Next-GenerationEU - National Recovery and Resilience Plan (NRRP) – MISSION 4 COMPONENT 2, INVESTMENT N. 1.4 – CUP N. I53C22000690001. 
Partial support for these authors was also provided by the INdAM - GNCS project ``Strutture di matrici e di funzioni per la sintesi di circuiti quantistici efficienti'', CUP E53C23001670001.
G. Del Corso is also partially supported by  the U.S. Department of Energy, Office of Science, National Quantum Information Science Research Centers, Superconducting Quantum Materials and Systems Center (SQMS), under Contract No. 89243024CSC000002.
}}

\author{Giacomo Antonioli\,\orcidlink{0000-}\thanks{Department of Computer Science, University of Pisa, Italy
(\email{giacomo.antonioli@phd.unipi.it}).}
\and Gianna M. Del Corso\,\orcidlink{0000-0002-5651-9368}\thanks{Department of Computer Science, University of Pisa, Italy
(\email{gianna.delcorso@unipi.it}).}
\and Roel Van Beeumen\,\orcidlink{0000-0003-2276-1153}\thanks{Applied Mathematics and Computational Research Division,
Lawrence Berkeley National Laboratory, Berkeley, CA 94720, USA}}

\begin{document}
\maketitle

\begin{abstract}
We present an explicit quantum circuit that constructs a block encoding of a
generator-representable semiseparable (one-pair) matrix directly from its two
generators. For an $N\times N$ semiseparable
matrix with $N=2^{n}$, the circuit uses $2n+2$ ancilla qubits and has
$O(\polylog N)$ depth, and it realizes a $(2\sqrt{N},\,2n+2,\,2 \sqrt{2}\ \varepsilon)$-block
encoding, where $\varepsilon$ is controlled by the accuracy of the
state-preparation routines that load the generators.  The building blocks of
the circuit are cyclic shifts, swaps, and Hadamard gates, each admitting an
efficient implementation through the binary decomposition of the relevant indices.
\end{abstract}



\section{Introduction}
\label{sec:intro}
Block encodings are the standard input model for quantum linear algebra. Indeed, 
quantum algorithms for solving linear systems, estimating spectral quantities,
or applying matrix functions all assume access to a unitary whose leading block
is a suitably scaled copy of the target matrix, and their efficiency is inherited
from the quality of that encoding, namely its ancilla count, its circuit depth, and the
subnormalization factor by which the matrix is scaled. Since most
algorithms simply presuppose such access, explicit and economical constructions
have drawn increasing attention. 
Structured matrices have long been a source of algorithmic gains in classical
numerical linear algebra: sparsity  patterns, displacement structure, banded-ness and off-diagonal
low-rank structure each translate into fast solvers, and it is natural to ask whether the same structure yields analogous savings on a quantum computer. For sparse matrices Camps et
al.~\cite{camps2024explicit} give circuits acting directly on the sparsity
pattern, while FABLE~\cite{fable} targets arbitrary dense matrices
through an approximate, compressible circuit.

For matrices with displacement structure this is by now well established. Efficient quantum circuits and block encodings have been designed for circulant, Toeplitz and Hankel matrices by exploiting their parametrization through a small number of
generators and their close relationship with the quantum Fourier transform
\cite{MW16,ZW17}; displacement-structured matrices have been
decomposed into linear combinations of shift matrices to yield
$\varepsilon$-approximate block encodings in both black-box and QRAM data-access models \cite{displacement21}, and more recent work builds block encodings directly from an arithmetic description of a matrix's sparsity and repeated-value
pattern, with Toeplitz matrices as a guiding example \cite{SCC24}. For rank-structured matrices, Nguyen, Kiani and Lloyd~\cite{NKL22}
block-encode hierarchical (HODLR/HSS) kernel matrices by recursively combining encodings of their low-rank blocks through linear combinations of unitaries,
attaining near-optimal linear-system runtimes for dense, full-rank kernels, and Tang and Lai~\cite{TL26} obtain a direct hierarchical encoding that lowers the required oracle access from $O(N^{2})$ to $O(N)$ entries.  The common thread across these results is that it is data sparsity, i.e., the number of parameters needed to describe the matrix, rather than the number of its nonzero entries that governs the achievable quantum cost, exactly as it governs classical complexity.

Generator-representable semiseparable matrices are a particularly clean instance of this principle: the entire $N\times N$ matrix is fixed by two length-$N$
vectors, giving a  data-sparse yet dense,  full-rank object.  Beyond their defining role as inverses of irreducible
tridiagonal matrices, they arise as the Green's functions of two-point boundary
value problems, as covariance matrices of Gauss--Markov and Brownian-motion type processes in statistics and filtering, in the oscillation theory of small mechanical vibrations, and in the theory of orthogonal polynomials \cite{gantmacher2002oscillation,vandebril2008matrix}. A first quantum block encoding tailored to this class was recently proposed
in~\cite{qseparxiv}, where $S$ is encoded through a triangular-diagonal
factorization; the resulting circuit has polylogarithmic depth, but the
factorization is paid for by a subnormalization factor and an error that both
grow as $O(N^{2})$. In this work we show that no factorization is needed.
Loading the two generators into separate registers already produces every
product $u_iv_j$, hence every entry of $S(u,v)$, and an explicit circuit of
cyclic shifts, swaps and Hadamard gates discards the spurious products and
routes each entry to its place in the leading block. The resulting encoding has
$O(\polylog N)$ depth and uses $2n+2$ ancillas instead of $2n+7$, while the
subnormalization drops from $O(N^{2})$ to $2\sqrt{N}$ and the error from
$O(N^{2}\varepsilon)$ to $2\sqrt{2}\,\varepsilon$, independently of $N$. Its only
matrix-dependent ingredients are the two state preparations, and it extends to
the unsymmetric case. A detailed comparison is given in
Section~\ref{sec:comparison}.

The paper is organized as follows. In Section~\ref{sec:prelim} we introduce the  notation and preliminary facts about quantum computation. In Section~\ref{sec:semisep} we review the definition of generator representable  semiseparable matrices, in both the symmetric and unsymmetric case. Section~\ref{sec:BB} introduces the building blocks of our construction, namely the cyclic shift and the composite operators derived from it, together with their circuits. Section~\ref{sec:BEsemisep} is devoted to the main result of this paper, the block encoding of symmetric, generator representable semiseparable matrices: there we present the main idea and the complete circuit, and then assemble the building blocks to prove the theorem. Section~\ref{sec:unsymm} extends the result to unsymmetric generator-representable matrices. 
 Section~\ref{sec:comparison}
compares our block encoding with the state of the art for relevant classes of
structured matrices. Numerical experiments are presented in
Section~\ref{sec:experiments}, and conclusions in Section~\ref{sec:concl}.

\section{Quantum Computing preliminaries}
\label{sec:prelim}

This section briefly recalls the qubit formalism, the gates used in our
circuits, and the formal notion of a block encoding, on which the
statement of Theorem~\ref{thm:main} relies. The material is entirely
standard and is included only to make the paper accessible to readers with
a numerical linear algebra background but little prior exposure to quantum
computing; a reader already familiar with these notions can skip directly
to Section~\ref{sec:semisep}.
\paragraph{Notation}
We adopt the bra-ket Dirac notation, writing $\bra{\cdot}$ and
$\ket{\cdot}$ for row and column vectors respectively. The two canonical
basis vectors of $\mathbb{C}^2$ are denoted $\ket{0} = [1,0]^T$ and
$\ket{1} = [0,1]^T$. A qubit is a unit vector in the two-dimensional
complex Hilbert space $\mathbb{C}^2$ with basis $\{\ket{0},\ket{1}\}$; its
most general state is a superposition $\alpha\ket{0}+\beta\ket{1}$ with
$|\alpha|^2+|\beta|^2=1$. The state of several qubits taken jointly is
described by the tensor product of the individual spaces, so that
$\ket{x}\ket{y} = \ket{xy} = \ket{x}\otimes\ket{y} \in
\mathbb{C}^2\otimes\mathbb{C}^2$. More generally, a
register of $n$ qubits lives in $(\mathbb{C}^2)^{\otimes n}=\mathbb{C}^N$
with $N=2^n$. This space has a natural basis, the computational
basis $\{\ket{j}\}_{j=0}^{N-1}$, where $j$ is the integer whose $n$-bit
binary expansion labels the basis state, and $\ket{j}$ coincides with the
$j$-th column of the identity matrix $I_N$. A general state of an
$n$-qubit register,
\[
  \ket{\psi} = \sum_{j=0}^{N-1}\alpha_j\ket{j}, \qquad \sum_j |\alpha_j|^2=1,
\]
thus simultaneously carries $2^n$ complex amplitudes.

\paragraph{Gates}
A quantum computation consists of applying a sequence of unitary matrices,
called {gates}, to a qubit register. The single-qubit gates used in
this paper are
\[
  H = \frac{1}{\sqrt{2}}\begin{pmatrix}1&1\\1&-1\end{pmatrix},\quad
  X = \begin{pmatrix}0&1\\1&0\end{pmatrix},\quad
  Z = \begin{pmatrix}1&0\\0&-1\end{pmatrix},\quad
  R_y(\theta) = \begin{pmatrix}\cos\frac{\theta}{2}&-\sin\frac{\theta}{2}\\
                               \sin\frac{\theta}{2}&\cos\frac{\theta}{2}\end{pmatrix}.
\]
Among two-qubit gates, we rely on the CNOT and SWAP gates,
\[
\operatorname{CNOT}=\begin{bmatrix}
    1 & 0&0&0\\0&1&0&0\\0&0&0&1\\0&0&1&0
\end{bmatrix}, \qquad
{\rm SWAP}=\begin{bmatrix}
    1 & 0&0&0\\0&0&1&0\\0&1&0&0\\0&0&0&1
\end{bmatrix}.
\]
The CNOT gate flips the second qubit whenever the first qubit is in state
$\ket{1}$, while the SWAP gate exchanges the states of two adjacent qubits.
More generally, a controlled-$U$ gate applies a unitary $U$ to a target register only when a designated
control qubit is in a specific state.
On a state vector, this means that $U$ is applied to the component of the state in which the control 
qubit is in the selected state, while the component in which the control qubit is in the complementary
state is left unchanged; as an operator, the gate is block-diagonal matrix with respect to the 
configuration induced  by the control qubit, with $U$ acting on one diagonal block and the identity on the other. 
Multi-controlled gates extend this idea by requiring several control qubits to be in a particular
configuration before the target unitary is applied: on a state vector, only the component corresponding to that configuration is transformed by $U$, while as an operator the gate is block-diagonal, acting as $U$ on the block selected by the configuration and as the identity on all remaining blocks.

We define the C-SWAP$(c,a,b)$ gate as an gate that swaps the qubits indexed by $a,b$ and is controlled by the state of the qubit $c$.

A quantum circuit is simply an ordered sequence of gates of this kind.

\paragraph{Block encodings}

In quantum computing all the valid operations require the operators to be unitary matrices. 
In classical numerical linear algebra many problems tackle general matrices, that do not respect this requirement. In order to extend the usage of quantum computing to all the linear algebra problems, we need to be able to represent non unitary matrices in these systems. 
A way to achieve this is through block encoding techniques: a target matrix
$A\in \mathbb{C}^{M\times N}$ is embedded as the top-left block of a
larger unitary matrix $U_A$, scaled by a factor $\alpha\ge \|A\|_2$, i.e.
$U_A$ encodes $A/\alpha$. We give below a formal definition.

\begin{definition}(Block encoding). Given an $n$-qubit matrix $A$ ($N = 2^n$), if we can find $\alpha, \varepsilon \in \mathbb{R}_+$ and an $(m+n)$-qubit
unitary matrix $U_A$ so that

$$
\| A - \alpha (\langle 0^m | \otimes I_N) U_A (|0^m \rangle \otimes I_N)\|_2 \le \varepsilon
$$
then $U_A$ is called an $(\alpha, m, \varepsilon)$-block-encoding of $A$. In particular, when the block encoding is exact with  $\varepsilon = 0$, $U_A$ is
called an $(\alpha, m)$-block-encoding of $A$.  \label{def:be}
\end{definition}

\section{Semiseparable Matrices}
\label{sec:semisep}
A matrix is called {semiseparable} when every submatrix extracted from its
lower triangular part, and every submatrix extracted from its upper triangular
part, has rank at most one, so that its off-diagonal structure is entirely
low-rank~\cite{eidelman2014separable,vandebril2008matrix}. Matrices of this kind,
and the broader family of rank-structured matrices to which they belong, arise
frequently in applications and have long been exploited to design fast algorithms
for solving linear systems, computing eigenvalues, and other fundamental tasks in
numerical linear algebra~\cite{vandebril2008matrix,eidelman2014separable,%
bebendorf2008hierarchical,ChanGu06,massei2020hm}. In this paper we focus on the
subclass in which the rank-one off-diagonal structure is described globally by
two vectors: the {generator-representable}, or {one-pair}, semiseparable
matrices, originally introduced in~\cite{gantmacher2002oscillation} in the context
of operator theory.

\begin{definition}\label{def:onepair}
Given vectors $u=(u_i)_{i=0,\ldots,N-1}$, $v=(v_j)_{j=0,\ldots,N-1}$ of dimension
$N$, the {one-pair} matrix associated with $u,v$ is the symmetric
$N\times N$ matrix $S=S(u,v)$ such that
$$
S_{i,j}=\begin{cases}
u_i v_j & \text{if } i\le j,\\
v_i u_j  & \text{if } i\ge j.
\end{cases}
$$
\end{definition}

In other words,
$$
S = S(u,v) = \left[ \begin{array}{ccccc}
u_0 v_0 & u_0 v_1 & u_0 v_2 & \cdots & u_0 v_{N-1} \\
v_1u_0  &  u_1 v_1& u_1 v_2 & \cdots & u_1 v_{N-1} \\
v_2u_0  &  v_2 u_1& u_2 v_2 & \cdots & u_2 v_{N-1} \\
\vdots  & \vdots  & \vdots  & \ddots & \vdots \\
v_{N-1}u_0  & v_{N-1}u_1  & v_{N-1}u_2  & \cdots &  u_{N-1}v_{N-1}
\end{array} \right].
$$

A well-known property of one-pair matrices is their connection to Jacobi
matrices, as stated in the following result~\cite{gantmacher2002oscillation,%
vandebril2005note}.

\begin{theorem}\label{thm:inv-tridiag}
Let $A$ be an $N\times N$ symmetric tridiagonal matrix of the form
$$
A=\begin{bmatrix}
a_0 & b_0 \\ b_0 & a_1 & b_1\\ & \ddots & \ddots & \ddots \\
&& b_{N-3} & a_{N-2}& b_{N-2}\\ &&& b_{N-2} & a_{N-1}
\end{bmatrix}
$$
with $b_i\neq 0$ for $i=0,\ldots,N-2$. Then the inverse of $A$ is a one-pair
matrix. Conversely, the inverse of an invertible one-pair matrix is an
irreducible symmetric tridiagonal matrix.
\end{theorem}

The unsymmetric analogue is obtained by allowing two independent generator pairs
for the two triangular parts; it arises, in particular, as the inverse of a
nonsymmetric irreducible tridiagonal matrix.

\begin{definition}\label{def:onepair-gen_unsymm}
Given four vectors $u,v,x,y\in\R^{N}$ with the diagonal compatibility
$u_iv_i=x_iy_i$ for $i=0,\dots,N-1$, the associated one-pair matrix
$S=S(u,v,x,y)$ is
$$
S_{ij}=\begin{cases} x_i y_j & i\le j,\\[2pt] u_j v_i & i\ge j. \end{cases}
$$
\end{definition}

The compatibility condition makes the two expressions agree on the diagonal,
and for $x=u$, $y=v$ one recovers the symmetric matrix $S(u,v)$ of
Definition~\ref{def:onepair}.

The duality between tridiagonal and generator-representable semiseparable
matrices~\cite{rozsa1994} admits an interesting generalization. A fundamental
result due to Asplund~\cite{asplund1959} establishes that the inverse of a
nonsingular band matrix of bandwidth $p$ is semiseparable of order $p$, in the
sense that every submatrix lying entirely above or below the main diagonal has
rank at most $p$. Equivalently, this can be expressed via the nullity
theorem~\cite{vandebrilnullity}: for an invertible matrix $A$, the rank of any
subblock of $A^{-1}$ is related to the nullity of a complementary subblock of
$A$, providing a precise quantitative link between the sparsity structure of a
matrix and the off-diagonal rank structure of its inverse. These results were
developed and extended to the block and nonsymmetric settings by several
authors~\cite{eidelmanGohberg1999,eidelmanGohberg2005block}, while the structural
properties of semiseparable matrices, including uniqueness of the generator
representation and behavior under unitary similarity transformations, are studied
in~\cite{bevilacqua2004}.

To place the class treated in this paper within the broader landscape of
rank-structured matrices, we recall the relevant hierarchy. A matrix is
{quasiseparable} (of order one) when the rank-at-most-one condition is
imposed only on its {strictly} triangular submatrices, leaving the diagonal
free; adding a free diagonal to a semiseparable matrix yields the
{semiseparable-plus-diagonal} matrices, which are quasiseparable but not, in
general, semiseparable. At the top of the hierarchy, {HODLR} (hierarchically
off-diagonal low-rank) and {HSS} (hierarchically semiseparable)
matrices~\cite{massei2020hm} relax the single-level, rank-one condition to
a bounded off-diagonal rank imposed recursively across a hierarchical partition
of the index set. One thus has the strict inclusions
\begin{equation}\label{eq:hierarchy}
\{\text{gen.-repr.}\}
\subsetneq \{\text{semisep.}\}
\subsetneq \{\text{semisep.}+\text{diag.}\}
\subsetneq \{\text{quasisep.}\}
\subsetneq \{\text{HODLR/HSS}\}.
\end{equation}
The first inclusion is strict because a diagonal matrix with  nonzero entries
is semiseparable but not generator-representable; the third because an irreducible
tridiagonal matrix is quasiseparable yet not semiseparable-plus-diagonal. The
construction of this paper targets the smallest class
in~\eqref{eq:hierarchy}, and Section~\ref{sec:concl} discusses its extension to
the larger ones.
\section{Building blocks}\label{sec:BB}
Before stating and proving our main result, we collect in this section the
operators from which the block encoding of Section~\ref{sec:BEsemisep} is assembled,
so that each of them can be introduced, and its circuit analyzed, on its own.
All of them are built from a single primitive, the cyclic
shift, raised to suitable powers. We first recall the shift and its circuit, then
state once the square-and-multiply principle used to realize all required powers,
and present the composite operators $Z_{\mathrm{blk}}/Z_{\mathrm{rblk}}$, $Z_{\mathrm{sel}}$, $Z_{\mathrm{fix}}$, and
$Z_{\mathrm{flip}}$, giving for each its definition and its circuit. Finally we present the $U_{\text{C-SWAP}}$ operator.
Each operator plays a precise role in the construction: $U_{\text{C-SWAP}}$ creates the
branch carrying the upper triangular part, $Z_{\mathrm{sel}}$ discards the spurious products,
$Z_{\mathrm{fix}}$ moves the surviving entries to their final rows, and $Z_{\mathrm{flip}}$ places
each column of $S$ in its final position. These roles are explained in
Section~\ref{sec:idea}, and the operators are applied in Section~\ref{sec:assembly}.

\subsection{The Cyclic shift and its circuit}\label{sec:shift}
The cyclic shift matrix $Z_N \in \mathbb{R}^{N \times N}$ is the fundamental
building block used throughout the construction of all subsequent operators.
It acts by cyclically permuting the canonical basis, while its adjoint
$Z_N^{\dagger}$ implements the inverse permutation. 
For the canonical basis $\{e_0,\ldots,e_{N-1}\}$, their action is
\begin{equation*}
Z_N e_j = e_{j-1 \bmod N},
\qquad
Z_N^{\dagger} e_j = e_{j+1 \bmod N}.
\end{equation*}
For instance, for $N=4$ and $x=[x_0,x_1,x_2,x_3]^T$ we have
\begin{equation*}
Z_4 x =
\begin{bmatrix}
0 & 1 & 0 & 0\\
0 & 0 & 1 & 0\\
0 & 0 & 0 & 1\\
1 & 0 & 0 & 0
\end{bmatrix}
\begin{bmatrix} x_0\\ x_1\\ x_2\\ x_3 \end{bmatrix}
=
\begin{bmatrix} x_1\\ x_2\\ x_3\\ x_0 \end{bmatrix},
\qquad
Z_4^{\dagger} x =
\begin{bmatrix}
0 & 0 & 0 & 1\\
1 & 0 & 0 & 0\\
0 & 1 & 0 & 0\\
0 & 0 & 1 & 0
\end{bmatrix}
\begin{bmatrix} x_0\\ x_1\\ x_2\\ x_3 \end{bmatrix}
=
\begin{bmatrix} x_3\\ x_0\\ x_1\\ x_2 \end{bmatrix},
\end{equation*}
so that $Z_N$ moves every entry of a vector one position up, and $Z_N^{\dagger}$ one
position down, the entry leaving at one end re-entering at the other.
Note that $Z_N^\dagger = Z_N^{N-1},$  $Z_N^N=I_N$, and that $Z_N^k$ performs a cyclic shift of $k$ positions. We will also need the following property, which for citation purposes we formulate as a technical lemma.

\begin{lemma}
\label{lem:tensor-power}
Let $D=2^{d}$ and let $Z_D$ be the cyclic shift on $\mathbb{R}^{D}$. For every
$\ell=0,\dots,d$,
\begin{equation}\label{eq:tensor-power}
  Z_D^{\,2^{\ell}} = Z_{D/2^{\ell}}\otimes I_{2^{\ell}},
\end{equation}
where the shifted factor acts on the $d-\ell$ most significant qubits and the
identity on the $\ell$ least significant ones. The same identity holds for the
adjoint, $(Z_D^{\dagger})^{2^{\ell}} = Z_{D/2^{\ell}}^{\dagger}\otimes I_{2^{\ell}}$.
\end{lemma}

For example, with $D=4$ and $\ell=1$, a shift by two positions acts on
pairs of consecutive entries as a whole:
\begin{equation*}
Z_4^{2} =
\begin{bmatrix}
0 & 0 & 1 & 0\\
0 & 0 & 0 & 1\\
1 & 0 & 0 & 0\\
0 & 1 & 0 & 0
\end{bmatrix}
=
\begin{bmatrix}
0 & I_2\\
I_2 & 0
\end{bmatrix}
= Z_2\otimes I_2,
\qquad
Z_4^{2}\,x = \begin{bmatrix} x_2\\ x_3\\ x_0\\ x_1 \end{bmatrix},
\end{equation*}
so that $Z_4^2$ reduces to the single-qubit shift $Z_2=X$ applied to the most
significant qubit.

Cyclic shift admit a natural circuit implementation as a cascade of multi-control\-led gates~\cite{camps2024explicit}. For instance, when $N=8$, a circuit with open
controls implements the map $x \mapsto x-1 \ \mathrm{mod}\ 8$, corresponding
to the operator $Z_8$. Replacing the open controls with closed controls
inverts the direction of the shift, yielding $x \mapsto x+1 \ \mathrm{mod}\ 8$
and thus implementing $Z_8^\dagger$. These two implementations are shown in
Figure~\ref{fig:shift-and-inverse-shift}.

\begin{figure}[h!]
\centering
\begin{minipage}{0.42\textwidth}
\centering
$
\begin{myqcircuit}
\lstick{q_{0}} & \targ      & \qw        & \qw      & \qw \\
\lstick{q_{1}} & \ctrlo{-1} & \targ      & \qw      & \qw \\
\lstick{q_{2}} & \ctrlo{-1} & \ctrlo{-1} & \gate{X} & \qw
\end{myqcircuit}
$
\caption*{$Z_8$: $x\mapsto x-1\ {\rm mod}\ 8$}
\end{minipage}
\begin{minipage}{0.42\textwidth}
\centering
$
\begin{myqcircuit}
\lstick{q_{0}} & \targ     & \qw       & \qw      & \qw \\
\lstick{q_{1}} & \ctrl{-1} & \targ     & \qw      & \qw \\
\lstick{q_{2}} & \ctrl{-1} & \ctrl{-1} & \gate{X} & \qw
\end{myqcircuit}
$
\caption*{$Z_8^\dagger$: $x\mapsto x+1\ {\rm mod}\ 8$}
\end{minipage}
\hfill

\caption{Circuit representation of $Z_8$ and $Z_8^\dagger$.}
\label{fig:shift-and-inverse-shift}
\end{figure}

\subsection{The square-and-multiply principle} \label{sec:s&m}
Each composite operator below is block-dia\-go\-nal, with the $j$-th block carrying a
power $Z^{j}$ of some cyclic shift $Z$ of dimension $N=2^n$ or $2N=2^{n+1}$, with $j=0, 1, \ldots, N-1$. Rather than
constructing each power independently, we generate it from the binary
decomposition of the index. Writing
$$ 
j=\sum_{\ell=0}^{n-1}b_{\ell}^{(j)}\,2^{\ell}=(b_{n-1}^{(j)} b_{n-2}^{(j)}\cdots b_{0}^{(j)})_2,\qquad b_{\ell}^{(j)}\in\{0,1\}.
$$
Here $b_\ell^{(j)}$ denotes the $\ell$-th binary digit of $j$, starting from the
least significant digit. Each binary digit determines whether the
corresponding power of the cyclic shift is applied: the least
significant digit selects the factor $Z$, the next digit selects
$Z^2$, and, in general, the $\ell$-th digit selects
$Z^{2^\ell}$. Equivalently,
$$
Z^j=\prod_{\ell=0}^{n-1}\left(Z^{2^\ell}\right)^{b_\ell^{(j)}}.
$$
For instance, for $N=4$ ($n=2$) and $Z=Z_4$, the indices $j=0,1,2,3$ have
binary expansions $(00)_2,(01)_2,(10)_2,(11)_2$, so that $Z^3=Z^2\,Z$, and the
block-diagonal operator carrying $Z^j$ in its $j$-th block factors as
\begin{equation*}
\begin{bmatrix}
I_4 &&&\\ & Z &&\\ && Z^{2} &\\ &&& Z^{3}
\end{bmatrix}
=
\underbrace{\begin{bmatrix}
I_4 &&&\\ & I_4 &&\\ && Z^{2} &\\ &&& Z^{2}
\end{bmatrix}}_{\text{controlled by } b_1}
\underbrace{\begin{bmatrix}
I_4 &&&\\ & Z &&\\ && I_4 &\\ &&& Z
\end{bmatrix}}_{\text{controlled by } b_0}.
\end{equation*}
The right factor applies $Z$ to the blocks with odd index, i.e.\ it is $Z$
controlled by the least significant index qubit, while the left factor applies
$Z^2$ to the blocks $j=2,3$, i.e.\ it is $Z^2$ controlled by the most
significant one. By Lemma~\ref{lem:tensor-power}, $Z_4^2=Z_2\otimes I_2$, so
the latter reduces to a single CNOT acting on the most significant qubit of the
target register. The same construction with $n$ factors yields all $N$ powers
$Z^0,\dots,Z^{N-1}$ at once.

Consequently each composite operator is realized as a product of $n=\log_2 N$
controlled, structured factors, the $\ell$-th controlled by the qubit carrying
$b_{\ell}$. We use this principle repeatedly without further comment.

\subsection{The block-shift operators $Z_{\mathrm{blk}}$ and $Z_{\mathrm{rblk}}$}\label{sec:zblk}
The block-shift operator $Z_{\mathrm{blk}}\in \mathbb{R}^{2N^2\times 2N^2}$ is a block-diagonal matrix, its $j$-th diagonal block being the power $Z_{2N}^j$ of the cyclic shift.  Explicitly,
\begin{equation*}
Z_{\mathrm{blk}} =
\begin{bmatrix}
I_{2N} & 0 & 0 & \cdots & 0 \\
0 & Z_{2N} & 0 & \cdots & 0 \\
0 & 0 & Z_{2N}^{2} & \cdots & 0 \\
\vdots & \vdots & \vdots & \ddots & \vdots \\
0 & 0 & 0 & \cdots & Z_{2N}^{N-1}
\end{bmatrix}
\in \mathbb{R}^{2N^2 \times 2N^2}.
\end{equation*}
Applied to a vector $Z_{\mathrm{blk}}$  shifts of $j$ positions the entries of the $j$-th block of size $2N$.  By the 
square-and-multiply principle~\ref{sec:s&m}, define for each digit $\ell$ the block-diagonal factor $Z_{\mathrm{blk}}^{(\ell)}$ that applies $Z_{2N}^{2^\ell}$ to the blocks whose $\ell$-th binary digit is one and the identity to the others, i.e. 
$$
Z_{\mathrm{blk}}^{(\ell)}
=
\begin{bmatrix}
Z_{2N}^{b_\ell^{(0)}2^\ell} \\
& Z_{2N}^{b_\ell^{(1)}2^\ell} \\
&& \ddots \\
&&& Z_{2N}^{b_\ell^{(N-1)}2^\ell}
\end{bmatrix}.
$$
Then 
$$
Z_{\mathrm{blk}} =Z_{\mathrm{blk}}^{(n-1)} Z_{\mathrm{blk}}^{(n-2)} \cdots Z_{\mathrm{blk}}^{(0)},
$$
a product of only $n=\log N$ structured block-diagonal matrices.   

The reverse block-shift operator $Z_{\mathrm{rblk}}$ has the same structure but with the diagonal blocks ordered by decreasing powers of $Z_{2N}^{\dagger}$. 
\begin{equation*}
Z_{\mathrm{rblk}}
=
\begin{bmatrix}
\left(Z_{2N}^{\dagger}\right)^{N} & 0 & 0 & \cdots & 0 \\
0 & \left(Z_{2N}^{\dagger}\right)^{N-1} & 0 & \cdots & 0 \\
0 & 0 & \left(Z_{2N}^{\dagger}\right)^{N-2} & \cdots & 0 \\
\vdots & \vdots & \vdots & \ddots & \vdots \\
0 & 0 & 0 & \cdots & Z_{2N}^{\dagger}
\end{bmatrix}
\in \mathbb{R}^{2N^2 \times 2N^2}.
\end{equation*}
While $Z_{2N}$ induces an upward
cyclic shift, its adjoint $Z_{2N}^{\dagger}$ implements the opposite
(downward) shift. As a result, $Z_{\mathrm{rblk}}$ reverses the direction of the block-wise
action of $Z_{\mathrm{blk}}$, producing a complementary block structure. 

Since $Z_{2N}^{N}=Z_{2N}^{-N}=\left(Z_{2N}^{\dagger}\right)^N$ and exponents are taken modulo $2N$, applying a half-period shift to every diagonal block of $Z_{\mathrm{blk}}$ turns $Z_{2N}^j$ into $ Z_{2N}^{N+j} = Z_{2N}^{-(N-j)} = \left(Z_{2N}^{\dagger}\right)^{N-j}.$
Hence $Z_{\mathrm{rblk}}$ need not to be build independently:
\begin{equation} \label{eq:zbi}
Z_{\mathrm{rblk}}=(I_N \otimes Z_{2N}^N) Z_{\mathrm{blk}}.
\end{equation}

\subsection{The operator $Z_{\mathrm{sel}}$ and its circuit}\label{sec:zsel}
The operator $Z_{\mathrm{sel}} \in \mathbb{R}^{4N^2 \times 4N^2}$ 
combines the two complementary block-shift structures $Z_{\mathrm{blk}}$ and $Z_{\mathrm{rblk}}$.
Its role is to assemble both operators into a single block-diagonal form and
to replicate this structure $N$ times via a Kronecker product with $I_N$,
so that each copy acts on a separate diagonal block. It is given by
\begin{equation*}
Z_{\mathrm{sel}}
=\begin{bmatrix}
Z_{\mathrm{blk}} & 0 \\
0 & Z_{\mathrm{rblk}}
\end{bmatrix}.
\end{equation*}
Using~\eqref{eq:zbi}, $Z_{\mathrm{sel}}$  admits the factorization
\begin{equation}
\label{eq:Zg_decomposition}
Z_{\mathrm{sel}}
=\left(
\begin{bmatrix}
I_{2N^2} & 0 \\
0 & I_N \otimes Z_{2N}^{N}
\end{bmatrix}
\left(I_2 \otimes Z_{\mathrm{blk}}\right)\right),
\end{equation}
which applies $I_2\otimes Z_{\mathrm{blk}}$ and then maps the second copy of $Z_{\mathrm{blk}}$ to $Z_{\mathrm{rblk}}$ via the block-selective half-period shift, avoiding the independent construction of $Z_{\mathrm{rblk}}$.  

\paragraph{Circuit}
The construction acts on a total of $2n+2$ qubits, partitioned into two  registers of $n+1$ qubits each, one used for indexing and the other for the shift itself. Because  the shifts act on a space of dimension $2N$, the shift register requires $n+1$ qubits. Moreover, because $Z_{\mathrm{sel}}$ ranges over 
all powers of $Z_{2N}$, from $Z_{2N}^0$ to
$Z_{2N}^{2N-1}$, the index register also requires $n+1$ qubits to encode the exponent. 
Following Lemma~\ref{lem:tensor-power} and the binary decomposition, the first controlled gate applies $Z_{2N}^N$ (realized as  $Z_2\otimes I_N$) , controlled by the leading index qubit. This gate accounts for the largest
binary contribution to the exponent of $Z_{2N}$ and distinguishes the lower
block from the upper one. The remaining controlled gates implement the
decreasing powers $Z_{2N}^{N/2},Z_{2N}^{N/4},\ldots,Z_{2N}$, realized
respectively by the shifts $Z_4,Z_8,\ldots,Z_{2N}$ on progressively larger
subregisters, together with identities on the remaining qubits. These gates
construct the common block-shift action $I_2\otimes Z_{\mathrm{blk}}$.
The complete construction is represented schematically in
Figure~\ref{fig:Zg_circ}. The depth of the circuit is $O(n^3).$

\begin{figure}[h!]
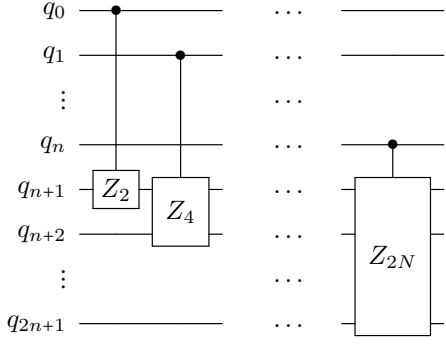

\centering
$
\begin{myqcircuit}
\lstick{q_{0}}      & \ctrl{4}       & \qw                 & \qw & &&&& \ldots &&&& & \qw                    & \qw \\
\lstick{q_{1}}    & \qw            & \ctrl{3}            & \qw & &&&& \ldots &&&& & \qw                    & \qw \\
\lstick{\vdots}     &                &                     &     & &&&& \ldots &&&& &                        &     \\
\lstick{q_{n}}     & \qw            & \qw                 & \qw & &&&& \ldots &&&& & \ctrl{1}               & \qw \\
\lstick{q_{n+1}}   & \gate{Z_2}     & \multigate{1}{Z_4}  & \qw & &&&& \ldots &&&& & \multigate{3}{Z_{2N}}  & \qw \\
\lstick{q_{n+2}}   & \qw            & \ghost{Z_4}         & \qw & &&&& \ldots &&&& & \ghost{Z_{2N}}         & \qw \\
\lstick{\vdots}     &                &                     &     & &&&& \ldots &&&& &                        &     \\
\lstick{q_{2n+1}}   & \qw            & \qw                 & \qw & &&&& \ldots &&&& & \ghost{Z_{2N}}         & \qw
\end{myqcircuit}$
\caption{Circuit for the construction of the operator $Z_{\mathrm{sel}}$.}
\label{fig:Zg_circ}
\end{figure}
\subsection{The operator $Z_{\mathrm{fix}}$ and its circuit}\label{sec:zfix}
The operator  $Z_{\mathrm{fix}}\in \mathbb{R}^{4N^3\times 4N^3}$,  extends the previously described shift constructions to a block
setting with an internal two-component structure. 
Its $j$-th block carries $(Z_N^{\dagger})^j$ on the leading component and the identity on the trailing one.
\begin{equation*}
Z_{\mathrm{fix}} =I_N\otimes I_2 \otimes 
\begin{bNiceMatrix}
\,I_N     &         &                    &        &        &                              &        \\
        & I_{N}   &                    &        &        &                              &        \\
        &         & Z_{N}^{\dagger}    &        &        &                              &        \\
        &         &                    & I_{N}  &        &                              &        \\
        &         &                    &        & \ddots &                              &        \\
        &         &                    &        &        & \left(Z_N^\dagger\right)^{N-1} &      \\
        &         &                    &        &        &                              & I_{N}\,\,  \\
\CodeAfter
\def\x{3}
\def\y{3}
  \tikz \draw[rounded corners=1pt, thin]
    ([xshift=-\x pt, yshift=\y pt]1-1.north west)
    rectangle
    ([xshift=5pt, yshift=-\y pt]2-2.south east);
  \tikz \draw[rounded corners=1pt]
    ([xshift=-\x pt, yshift=\y pt]3-3.north west)
    rectangle
    ([xshift=\x pt, yshift=-\y pt]4-4.south east);
  \tikz \draw[rounded corners=1pt]
    ([xshift=-\x pt, yshift=\y pt]6-6.north west)
    rectangle
    ([xshift=\x pt, yshift=-\y pt]7-7.south east);
\end{bNiceMatrix}\in\mathbb{R}^{4N^3\times4N^3}.
\end{equation*}

More explicitly, the nontrivial part of the $j$-th block is determined by
the block index. The leading component carries the power
$\left(Z_N^{\dagger}\right)^j$, whereas the trailing component remains equal
to $I_N$. Therefore, the block structure can be interpreted as
\[
Z_{\mathrm{fix}}^{(j)}=
\begin{bmatrix}
    \left(Z_N^{\dagger}\right)^j & 0\\
    0 & I_N
\end{bmatrix},
\]
with the appropriate value of $j$ selected by the block index.
As for the operator $Z_{\mathrm{blk}}$, the powers of $Z_N^{\dagger}$ can be generated
using the binary decomposition of the index. If $N=2^n$ and
\[
j=\sum_{\ell=0}^{n-1} b_\ell^{(j)}2^\ell,
\qquad b_\ell^{(j)}\in\{0,1\},
\]
then
\[
\left(Z_N^{\dagger}\right)^j
=
\prod_{\ell=0}^{n-1}
\left(\left(Z_N^{\dagger}\right)^{2^\ell}\right)^{b_\ell^{(j)}}.
\]
Each binary digit determines whether the corresponding power of
$Z_N^{\dagger}$ is applied to the leading component of the block. The second
component is not affected by these operations and remains the identity.
Hence, $Z_{\mathrm{fix}}$ is obtained by applying the same power-and-multiply construction
used for $Z_{\mathrm{blk}}$, with the additional restriction that the controlled inverse
shifts act only on the leading subblock of each diagonal block.

\paragraph{Circuit}

We now describe the quantum circuit that implements $Z_{\mathrm{fix}}$. The construction
uses $2n+2$ qubits. The first $n$ qubits form an index register that encodes
the integer $j$ in binary. An additional qubit acts as a control that enables
or disables the shift operations, and the remaining $n$ qubits form the
target register on which the operators $Z_N^\dagger$ act. Since $N=2^n$, the
target space is represented using $n$ qubits.

As before, the power of the shift is generated from the binary decomposition
of the index. The controlled gates implement the powers
\[
\left(Z_N^\dagger\right)^{N/2},
\left(Z_N^\dagger\right)^{N/4},
\ldots,
Z_N^\dagger,
\]
which are realized on the full target register as
\[
Z_2^\dagger\otimes I_{N/2}, \quad
Z_4^\dagger\otimes I_{N/4}, \quad
\ldots,\quad
Z_N^\dagger.
\]
Each qubit of the index register selects one of these powers. Their product
therefore gives the required downward shift
$\left(Z_N^\dagger\right)^j$ on the leading component of the $j$-th block.

The additional selection qubit prevents the same shifts from being applied to
the second component of the block. When this qubit selects the leading
component, the controlled downward shifts are applied; otherwise, the target
register is left unchanged.
The complete construction is represented schematically in
Figure~\ref{fig:Zs_circ}.

\begin{figure}[h!]
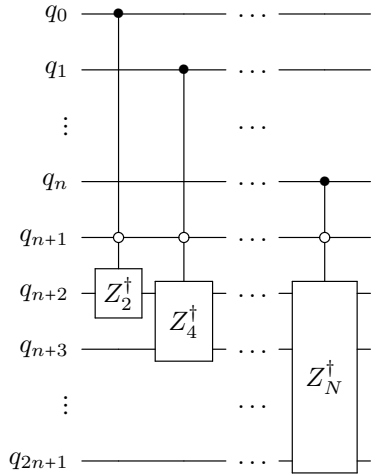

\centering
$
\begin{myqcircuit}
\lstick{q_{0}}   &\ctrl{4}  & \qw            & \qw & & \ldots && & \qw                           & \qw \\
\lstick{q_{1}}   & \qw        & \ctrl{3}       & \qw & &\ldots && & \qw                           & \qw \\
\lstick{\vdots}  &            &                &     & & \ldots && &                               &     \\
\lstick{q_{n}} & \qw        & \qw            & \qw & & \ldots && & \ctrl{1}                      & \qw \\
\lstick{q_{n+1}}   & \ctrlo{1} &\ctrlo{1} & \qw & & \ldots && & \ctrlo{1} & \qw \\
\lstick{q_{n+2}} & \gate{Z_2^{\dagger}}        &  \multigate{1}{Z_4^{\dagger}}       & \qw & & \ldots && & \multigate{3}{Z_{N}^{\dagger}}        & \qw \\
\lstick{q_{n+3}} & \qw                           & \ghost{Z_4^{\dagger}}        & \qw & & \ldots && & \ghost{Z_{N}^{\dagger}}        & \qw \\
\lstick{\vdots}  &                              &                               &     & & \ldots && &                                      &     \\
\lstick{q_{2n+1}}  & \qw        & \qw            & \qw & & \ldots && & \ghost{Z_{N}^{\dagger}}        & \qw
\end{myqcircuit}$

\caption{Circuit for the construction of the operator $Z_{\mathrm{fix}}$. The index
register controls the powers of $Z_N^\dagger$ selected by the binary
decomposition of the block index, while the additional selection qubit
restricts their action to the leading component of each diagonal block.}
\label{fig:Zs_circ}
\end{figure}

\subsection{The operator $Z_{\mathrm{flip}}$ and its circuit}\label{sec:zflip}
The block-shift operator $Z_{\mathrm{flip}}$ is a block-diagonal matrix whose
structure differs from the previously introduced shift-based
operators. In this case, the diagonal blocks are constructed
directly from tensor products of $2 \times 2$ Pauli-$X$ matrices
and identity operators, without relying on cyclic shift matrices.
Each block is determined by a binary index and encodes a
corresponding tensor product of $X^0 = I$ and $X^1 = X$ factors,
followed by an identity acting on the remaining degrees of
freedom:
\begin{equation*}    
Z_{\mathrm{flip}} = \begin{bNiceMatrix}
    X^0 \otimes \dots \otimes X^0 & \cdots & 0 & \cdots & 0 \\
   \vdots&\ddots  & 0 & \cdots & 0 \\
    0 & \cdots &X^{i_{n-1}} \otimes \dots \otimes X^{i_0} & \cdots & 0 \\
    \vdots & \cdots &\vdots & \ddots & \vdots \\
    0 & \cdots &0 & \cdots & X^1 \otimes \dots \otimes X^1 
\end{bNiceMatrix}\otimes I_{2N} \in\mathbb{R}^{4N^3\times4N^3}.
\end{equation*}
Each diagonal block of $Z_{\mathrm{flip}}$ can be described as:
\begin{equation*}
    (Z_{\mathrm{flip}})_i = P_i \otimes I_{2N},
    \qquad
    P_i := \bigotimes_{k=0}^{n-1} X^{i_{n-1-k}}
         = X^{i_{n-1}} \otimes \cdots \otimes X^{i_0},
 \end{equation*}
where $i = \sum_{\ell=0}^{n-1} i_\ell\, 2^\ell$ with $i_\ell\in\{0,1\}$ is the binary
decomposition of the  block index $i$, and $X^0 = I$, $X^1 = X$. 
In each block, the operator $P_i$ applies Pauli-$X$ on the $\ell$-th factor
whenever $i_\ell = 1$, and acts as the identity otherwise.
In other words, $Z_{\mathrm{flip}}$ applies a block-dependent pattern of Pauli-$X$ flips determined 
by the binary representation of the block index, while acting as the identity on the remaining subsystem.
We now describe the circuit implementation of $Z_{\mathrm{flip}}$.
The construction acts on $2n+2$ qubits, where the first $n$
qubits encode the binary representation of the index and act as
controls. One additional qubit is unused, while the next $n$
qubits form the target registers on which the permutations are
applied.

\paragraph{Circuit} The circuit is built as a sequence of independent controlled-NOT
operations implementing the block-dependent bit-flip structure:
 
\begin{equation*}
    Z_{\mathrm{flip}}=\left(\prod_{a=0}^{n-1}I_{2^a}\otimes \text{CNOT}(a,a+n{+1})\otimes I_{2^{n-a-1}}\right).
\end{equation*}
 The complete construction is represented schematically in
Figure~\ref{fig:Zr_circ}.

\begin{figure}[h!]
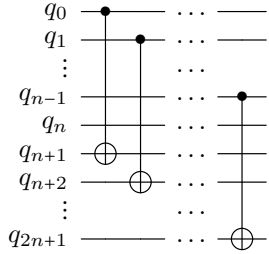

\centering
$
\def\n{5}
\begin{myqcircuit}
\lstick{q_{0}}   & \ctrl{\n+1} & \qw  & \qw        & & \ldots && &\qw & \qw \\
\lstick{q_{1}}   & \qw        & \ctrl{\n+1}& \qw   & & \ldots && &\qw & \qw \\
\lstick{\vdots}  &            &           & & & \ldots && &   &\\
\lstick{q_{n-1}} & \qw        & \qw     & \qw     & & \ldots & & &\ctrl{\n+1} & \qw \\
\lstick{q_{n}}   & \qw        & \qw     & \qw     & & \ldots && &\qw & \qw \\
\lstick{q_{n+1}} & \targ      & \qw    & \qw     & & \ldots && &\qw & \qw \\
\lstick{q_{n+2}} & \qw        & \targ   & \qw     & & \ldots && & \qw & \qw \\
\lstick{\vdots}  &            &         &  & & \ldots && &     &     \\
\lstick{q_{2n+1}}  & \qw        & \qw   & \qw     & & \ldots && & \targ & \qw \\
\end{myqcircuit}
$
\caption{Circuit for the construction of the operator $Z_{\mathrm{flip}}$.}
\label{fig:Zr_circ}
\end{figure}

\subsection{The $U_{\text{C-SWAP}}$ operator}\label{sec:cswap}
The last building block is the controlled-swap operator
$U_{\text{C-SWAP}}$, whose role is to exploit the symmetry of $S$: within
the branch selected by the triangular qubit, it exchanges the contents of
the two generator registers, turning the blocks $u_j v$ into $v_j u$. It is defined as
\begin{align*}
U_{\text{C-SWAP}} = I_N\otimes \left(\prod_{a=0}^{n}I_{2^a}\otimes
    \text{C-SWAP}(n,n+1+a,2n+2+a)\otimes I_{2^{n-a-1}}\right),
\end{align*}
where $\text{C-SWAP}(c,a,b)$ is the controlled-swap gate introduced in
Section~\ref{sec:prelim}, which swaps the qubits $a$ and $b$ conditionally on the state of qubit $c$.

\begin{rem}\label{rem:cost}
Let $N=2^n$. The operators of Section~\ref{sec:BB} can be realized by circuits with the following depth:
\begin{enumerate}
\item the cyclic shift $Z_{2^k}$ of Figure~\ref{fig:shift-and-inverse-shift} has depth $O(k^2)$;
\item $Z_{\mathrm{sel}}$, a product of $n+1$ singly controlled shifts
      $Z_{2^k}$, $k=1,\dots,n+1$, has depth $O(n^3)$;   
\item $Z_{\mathrm{fix}}$, a product of $n$ doubly controlled inverse shifts
      $Z^\dagger_{2^k}$, $k=1,\dots,n$, has depth $O(n^3)$;
\item $Z_{\mathrm{flip}}$ consists of $n$ CNOT gates acting on pairwise
      disjoint qubit pairs and has depth $O(1)$;
\item $U_{\text{\rm C-SWAP}}$ consists of $n$ C-SWAP gates sharing the
      same control qubit, and has depth $O(n)$.
\end{enumerate}
Consequently, the circuit of Theorem~\ref{thm:main} in Section~\ref{sec:BEsemisep}, which combines these operators with two state preparations, has depth
$O(\log^3 N) + D_{\mathrm{sp}}$, where $D_{\mathrm{sp}}$ denotes the depth
of the state-preparation routine.

The cost can be lowered using QFT-based algorithm for implementing cyclic shifts, but the overall cost remains polylogarithmic in $N$.
\end{rem}

\section{Block encoding of semiseparable matrices} \label{sec:BEsemisep}
We now state our main result. Section~\ref{sec:idea} gives the high-level idea of the construction together with the complete circuit, and Section~\ref{sec:assembly}
assembles the building blocks of Section~\ref{sec:BB} and proves correctness, including the error bound.

\begin{theorem}[Block encoding of a semiseparable matrix]
\label{thm:main}
Let $N=2^{n}$ and let $S(u,v)\in\R^{N\times N}$ be the semiseparable matrix with generators $u,v\in\R^{N}$. Assume that
\begin{enumerate}
  \item the generators are unit-norm, $\norm{u}_2=\norm{v}_2=1$;
  \item the state-preparation routines load unit-norm approximations
        $\tilde u,\tilde v$ with $\norm{u-\tilde u}_2\le\varepsilon$ and
        $\norm{v-\tilde v}_2\le\varepsilon$;
  \item the state preparation routines  have circuit depth of $O(\polylog N)$.
\end{enumerate}
Then there exists a quantum circuit acting on $3n+2$ qubits, with circuit depth of
$O(\polylog N)$, that implements a
$\bigl(2\sqrt N,\,2n+2,\,2\sqrt2\,\varepsilon\bigr)$-block-encoding of $S(u,v)$.
\end{theorem}

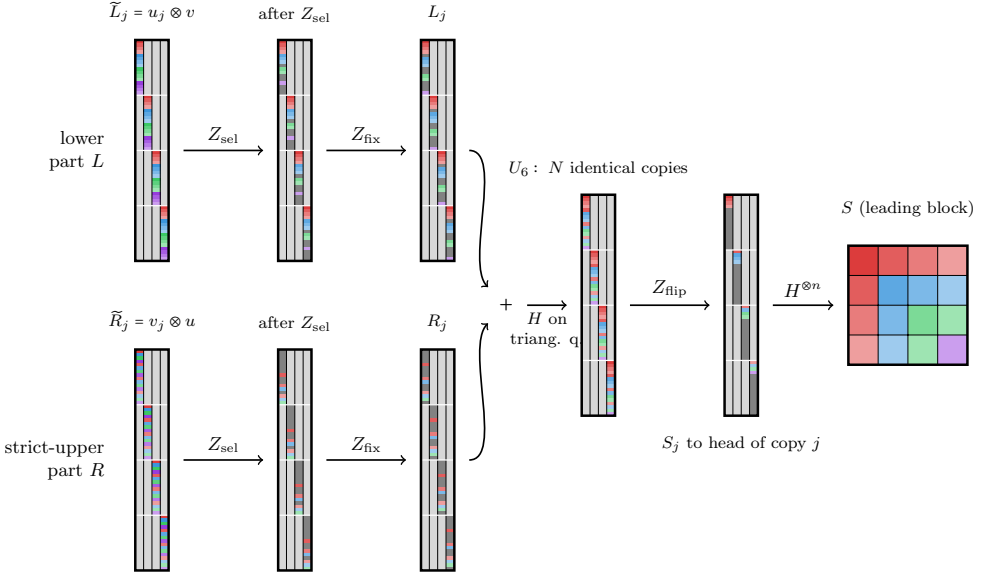
\begin{figure}[!htbp]
\centering
\adjustbox{max width=\textwidth,max totalheight=0.92\textheight}{%
\begin{tikzpicture}[op/.style={->,thick,shorten >=3pt,shorten <=3pt}]
  \def\dx{2.4}\def\yL{0}\def\yR{-5.2}\def\yM{-2.6}

  \node (t1) at (0,\yL) {\opBD{
    col1/100,col1/83,col1/67,col1/50, col2/100,col2/83,col2/67,col2/50,
    col3/100,col3/83,col3/67,col3/50, col4/100,col4/83,col4/67,col4/50}};
  \node (t2) at (\dx,\yL) {\opBD{
    col1/100,col1/83,col1/67,col1/50, col2/83,col2/67,col2/50,gg/100,
    col3/67,col3/50,gg/100,gg/100, col4/50,gg/100,gg/100,gg/100}};
  \node (t3) at (2*\dx,\yL) {\opBD{
    col1/100,col1/83,col1/67,col1/50, gg/100,col2/83,col2/67,col2/50,
    gg/100,gg/100,col3/67,col3/50, gg/100,gg/100,gg/100,col4/50}};
  \draw[op] (t1) -- node[above,font=\small]{$Z_{\mathrm{sel}}$} (t2);
  \draw[op] (t2) -- node[above,font=\small]{$Z_{\mathrm{fix}}$} (t3);
  \node[above=1pt of t1,font=\footnotesize]{$\Lt_j=u_j\otimes v$};
  \node[above=1pt of t2,font=\footnotesize]{after $Z_{\mathrm{sel}}$};
  \node[above=1pt of t3,font=\footnotesize]{$L_j$};
  \node[left=0.25cm of t1,font=\small,align=right]{lower\\ part $L$};

  \node (r1) at (0,\yR) {\opBD{
    col1/100,col2/100,col3/100,col4/100, col1/83,col2/83,col3/83,col4/83,
    col1/67,col2/67,col3/67,col4/67, col1/50,col2/50,col3/50,col4/50}};
  \node (r2) at (\dx,\yR) {\opBD{
    gg/100,gg/100,gg/100,gg/100, gg/100,gg/100,gg/100,col1/83,
    gg/100,gg/100,col1/67,col2/67, gg/100,col1/50,col2/50,col3/50}};
  \node (r3) at (2*\dx,\yR) {\opBD{
    gg/100,gg/100,gg/100,gg/100, col1/83,gg/100,gg/100,gg/100,
    col1/67,col2/67,gg/100,gg/100, col1/50,col2/50,col3/50,gg/100}};
  \draw[op] (r1) -- node[above,font=\small]{$Z_{\mathrm{sel}}$} (r2);
  \draw[op] (r2) -- node[above,font=\small]{$Z_{\mathrm{fix}}$} (r3);
  \node[above=1pt of r1,font=\footnotesize]{$\Rt_j=v_j\otimes u$};
  \node[above=1pt of r2,font=\footnotesize]{after $Z_{\mathrm{sel}}$};
  \node[above=1pt of r3,font=\footnotesize]{$R_j$};
  \node[left=0.25cm of r1,font=\small,align=right]{strict-upper\\ part $R$};

  \node (plus) at (2*\dx+1.15,\yM) {\large $+$};
  \node (u6) at (2*\dx+2.7,\yM) {\opBD{
    col1/100,col1/83,col1/67,col1/50, col1/83,col2/83,col2/67,col2/50,
    col1/67,col2/67,col3/67,col3/50, col1/50,col2/50,col3/50,col4/50}};
  \node (zf) at (2*\dx+2.7+\dx,\yM) {\opBDfour
    {col1/100,col1/83,col1/67,col1/50, gg/100,gg/100,gg/100,gg/100, gg/100,gg/100,gg/100,gg/100, gg/100,gg/100,gg/100,gg/100}
    {col1/83,col2/83,col2/67,col2/50, gg/100,gg/100,gg/100,gg/100, gg/100,gg/100,gg/100,gg/100, gg/100,gg/100,gg/100,gg/100}
    {col1/67,col2/67,col3/67,col3/50, gg/100,gg/100,gg/100,gg/100, gg/100,gg/100,gg/100,gg/100, gg/100,gg/100,gg/100,gg/100}
    {col1/50,col2/50,col3/50,col4/50, gg/100,gg/100,gg/100,gg/100, gg/100,gg/100,gg/100,gg/100, gg/100,gg/100,gg/100,gg/100}};
  \node (sm) at (2*\dx+2.7+2*\dx+0.4,\yM) {\Smatrix};

  \draw[op] (u6) -- node[above,font=\small]{$Z_{\mathrm{flip}}$} (zf);
  \draw[op] (zf) -- node[above,font=\small]{$H^{\otimes n}$} (sm);
  \node[above=1pt of u6,font=\footnotesize]{$U_6:\ N$ identical copies};
  \node[below=2pt of zf,font=\footnotesize]{$S_j$ to head of copy $j$};
  \node[above=6pt of sm,font=\footnotesize]{$S$ (leading block)};

  \draw[op] (t3.east) .. controls +(0.7,0) and +(-0.7,0.7) .. (plus);
  \draw[op] (r3.east) .. controls +(0.7,0) and +(-0.7,-0.7) .. (plus);
  \draw[op] (plus) -- node[below,font=\footnotesize,align=center]{$H$ on\\ triang. q.} (u6);
\end{tikzpicture}}

\caption{Construction of the block encoding of a $4\times4$ semiseparable matrix,
shown on its first $N=4$ columns. Colour encodes the $u$-index $a=\min(i,j)$
and shade the $v$-index $b=\max(i,j)$. Left, top and bottom rows: from
$\Lt_j=u_j\otimes v$ and $\Rt_j=v_j\otimes u$, the shifts $Z_{\mathrm{sel}}$ and
$Z_{\mathrm{fix}}$ produce the lower/upper triangular columns $L_j$, $R_j$ (background gray $=$ zero,
faded $=$ discarded amplitudes). Right column: the Hadamard on the triangular qubit adds
the branches, so each of the $N$ identical copies now stacks all columns
$S_0,\dots,S_{N-1}$ ($U_6$); $Z_{\mathrm{flip}}$ moves $S_j$ to the head of the
$j$-th copy (the four columns become distinct; faded $=$ amplitudes later
discarded); finally $H^{\otimes n}$ superposes the copies and places $S$ in the
leading block.} \label{fig:algo}
\end{figure}

\newcommand{\gateUu}{\multigate{2}{U_u}}
\newcommand{\ghostUu}{\ghost{U_u}}

\newcommand{\gateUv}{\multigate{2}{U_v}}
\newcommand{\ghostUv}{\ghost{U_v}}

\newcommand{\gateZflip}{\multigate{6}{\mathrm{Z_{flip}}}}
\newcommand{\ghostZflip}{\ghost{\mathrm{Z_{flip}}}}

\newcommand{\gateZext}{\multigate{7}{\mathrm{Z_{sel}}}}
\newcommand{\ghostZext}{\ghost{\mathrm{Z_{sel}}}}

\newcommand{\gateZsel}{\multigate{6}{\mathrm{Z_{fix}}}}
\newcommand{\ghostZsel}{\ghost{\mathrm{Z_{fix}}}}
\newcommand{\barr}{\qw\barrier[-0.5em]{10}}

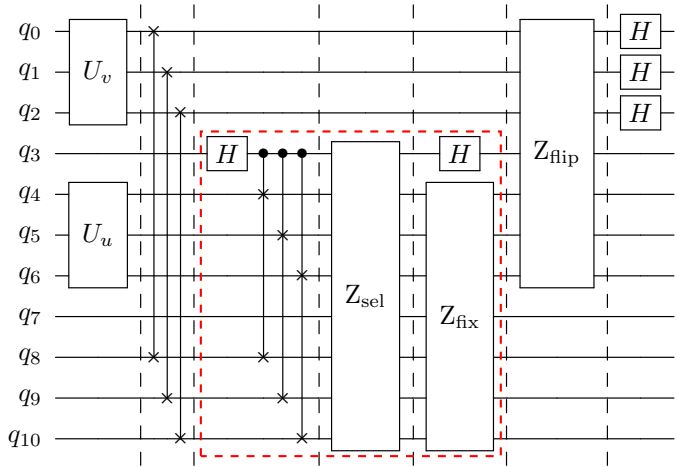
\begin{figure}[htbp]
    \centering
    \begin{tikzpicture}
        \node[inner sep=0, outer sep=0] (circ) {%
        $\begin{myqcircuit}
        \lstick{q_{0}}  & \gateUv  & \qw\barrier[-0.5em]{10}& \qswap\qwx[8] & \qw & \qw & \barrier[-1.25em]{10}\qw & \qw      & \qw      & \qw      & \qw      & \qw\barrier[-1.75em]{10} & \qw          & \qw\barrier[-1.75em]{10} & \qw          & \qw\barrier[-1.9em]{10} & \gateZflip  & \qw\barrier[-1.25em]{10} & \gate{H} & \qw \\
        \lstick{q_{1}}  & \ghostUv & \qw   & \qw           & \qswap\qwx[8] & \qw & \qw & \qw    & \qw      & \qw      & \qw      & \qw   & \qw          & \qw   & \qw          & \qw   & \ghostZflip & \qw   & \gate{H} & \qw \\
        \lstick{q_{2}}  & \ghostUv & \qw   & \qw           & \qw & \qswap\qwx[8] & \qw & \qw    & \qw      & \qw      & \qw      & \qw   & \qw          & \qw   & \qw          & \qw   & \ghostZflip & \qw   & \gate{H} & \qw \\
        \lstick{q_{3}}  & \qw      & \qw   & \qw           & \qw & \qw & \qw & \gate{H} & \ctrl{1} & \ctrl{2} & \ctrl{3} & \qw   & \gateZext    & \qw   & \gate{H}     & \qw   & \ghostZflip & \qw   & \qw      & \qw \\
        \lstick{q_{4}}  & \gateUu  & \qw   & \qw           & \qw & \qw & \qw & \qw      & \qswap\qwx[4] & \qw & \qw      & \qw   & \ghostZext   & \qw   & \gateZsel    & \qw   & \ghostZflip & \qw   & \qw      & \qw \\
        \lstick{q_{5}}  & \ghostUu & \qw   & \qw           & \qw & \qw & \qw & \qw      & \qw      & \qswap\qwx[4] & \qw & \qw   & \ghostZext   & \qw   & \ghostZsel   & \qw   & \ghostZflip & \qw   & \qw      & \qw \\
        \lstick{q_{6}}  & \ghostUu & \qw   & \qw           & \qw & \qw & \qw & \qw      & \qw      & \qw      & \qswap\qwx[4] & \qw & \ghostZext   & \qw   & \ghostZsel   & \qw   & \ghostZflip & \qw   & \qw      & \qw \\
        \lstick{q_{7}}  & \qw      & \qw   & \qw           & \qw & \qw & \qw & \qw      & \qw      & \qw      & \qw      & \qw   & \ghostZext   & \qw   & \ghostZsel   & \qw   & \qw         & \qw   & \qw      & \qw \\
        \lstick{q_{8}}  & \qw      & \qw   & \qswap        & \qw & \qw & \qw & \qw      & \qswap   & \qw      & \qw      & \qw   & \ghostZext   & \qw   & \ghostZsel   & \qw   & \qw         & \qw   & \qw      & \qw \\
        \lstick{q_{9}}  & \qw      & \qw   & \qw           & \qswap & \qw & \qw & \qw    & \qw      & \qswap   & \qw      & \qw   & \ghostZext   & \qw   & \ghostZsel   & \qw   & \qw         & \qw   & \qw      & \qw \\
        \lstick{q_{10}} & \qw      & \qw   & \qw           & \qw & \qswap & \qw & \qw    & \qw      & \qw      & \qswap   & \qw   & \ghostZext   & \qw   & \ghostZsel   & \qw   & \qw         & \qw   & \qw      & \qw \\
        \end{myqcircuit}$%
        };

        \def\padL{0pt}
        \def\padR{0pt}
        \def\padT{0pt}
        \def\padB{0pt}

        \begin{scope}[
            shift={($(circ.south west)+(\padL,\padB)$)},
            x={($(circ.south east)-(\padR,0)-(circ.south west)-(\padL,0)$)},
            y={($(circ.north west)-(0,\padT)-(circ.south west)-(0,\padB)$)}
        ]
            \draw[red, dashed, thick] (0.235,0.02) rectangle (0.72,0.725);
        \end{scope}
    \end{tikzpicture}
    \caption{The circuit for the block encoding of a symmetric semiseparable matrix $S(u, v)$}
    \label{fig:circ}
\end{figure}

\subsection{Idea of the construction} \label{sec:idea}
The starting point is to load the two generators $u$ and $v$ as quantum states via a state-preparation routine. Loaded in separated registers in a suitably entangled state, they produce all products $u_iv_j$. These contain the entries of $S(u, v)$ which selects $u_iv_j$ on the upper triangular part (for $i\le j$) and $u_j v_i$ for the lower triangular part (for $i\ge j$), together with spurious products that do not appear in $S$.
The task of the circuit is therefore twofold: discard the spurious
products, and route each surviving entry into the appropriate column of the
unitary realizing the block encoding.

This is accomplished as follows. A Hadamard gate on an ancillary qubit first duplicates the loaded amplitudes into two branches, intended to carry the lower and the strict upper triangular part of
$S$, respectively; a controlled swap of the two generator registers,
conditioned on the branching qubit, exploits the symmetry of $S$ to turn the
second branch into the reversed products $v_ju_i$. At this stage each branch
contains the correct entries of its triangular part interleaved with spurious
ones. A first shift operator, acting on a doubled register whose extra
$\ket{0}$ components serve as placeholders, moves the spurious entries out of
the way, replacing them with zeros; a second shift operator then repositions
the surviving entries so that each block carries a triangular column of $S$
in its final layout. A second Hadamard gate on the branching qubit interferes
the two branches, summing the lower and upper contributions and producing the
full columns of $S$, identically replicated across the copies indexed by a
further $n$-qubit register. A third shift operator permutes the entries
within each copy so that the $j$-th copy carries the $j$-th column of $S$ at
its top, and a final layer of Hadamard gates on the leading register places
the columns side by side, leaving $S/\alpha$ in the top-left block of the
resulting unitary.

Since all building blocks are shifts, swaps, and Hadamard gates, the whole
construction has $O(\mathrm{polylog}\,N)$ depth, assuming efficient state
preparation of the generators.

The action of the different shifts on the unitary operator is depicted in Figure~\ref{fig:algo}, while the complete circuit is shown in Figure~\ref{fig:circ}.

\subsection{Efficient state preparation}\label{sec:stateprep}
The polylogarithmic depth of the circuit stated in Theorem~\ref{thm:main} is based on the assumption that the generators $u$ and $v$ can be prepared by a state preparation routine whose circuit has a polylogarithmic depth. 
This holds either when the data are stored in a QRAM~\cite{giovannetti2008quantum,kerenidis2017quantum}, or  when the
generators are samples of a function. In the latter case the
entries are never computed one by one. Let $u_i=g(x_i)$ on the uniform grid
$x_i=i/N$ of $[0,1]$, and let $P$ be a polynomial of degree $d_\varepsilon$
with $\|P-g\|_\infty\le\varepsilon\|g\|_\infty$. Since $x_i$ is a linear function of the binary digits of $i$,
$x_i=\sum_{\ell=0}^{n-1}2^{\ell}b_\ell^{(i)}/N$, the diagonal matrix of the
grid variable splits as
\[
  D=\operatorname{diag}(x_i)=\sum_{\ell=0}^{n-1}\frac{2^{\ell}}{N}\,\Pi_\ell,
  \qquad
  \Pi_\ell = I_{2^{\,n-1-\ell}}\otimes
           \begin{bmatrix}0&0\\0&1\end{bmatrix}\otimes I_{2^{\ell}},
=\operatorname{diag}\bigl(b_\ell^{(i)}\bigr)
\]
where $\Pi_\ell$ is the projector onto the states whose $\ell$-th digit is one,
a single-qubit operator acting on the $\ell$-th qubit. A linear combination of
these $n$ commuting terms gives a block encoding of $D$ of depth $O(n)$, whose
subnormalization $\sum_{\ell}2^{\ell}/N=(N-1)/N$ equals $\|D\|_2$. Quantum singular value transformation builds $P(D)$
with $d_\varepsilon$ applications of that encoding, and the normalized state
$\ket u$ is obtained by applying $P(D)$ to the uniform superposition
$H^{\otimes n}\ket 0$ and amplifying the successful branch. The only classical
work is the computation of the $d_\varepsilon+1$ Chebyshev coefficients of $g$,
whose cost does not depend on $N$. The depth of the circuit is polylogarithmic if two conditions  hold at the same time:
$d_\varepsilon$ has to grow at most polylogarithmically, which is the case
with $d_\varepsilon=O(\log(1/\varepsilon))$ whenever $g$ is analytic on
$[0,1]$, and the maximum modulus of $g$ must be comparable to its root mean square over
the grid, this ratio being the number of amplification rounds
required~\cite{mcardle2025}.

The three-point discretization of a Sturm-Liouville operator gives a family of
tridiagonal matrices whose inverses are one-pair matrices, with generators that
are available in closed form and satisfy both conditions.
Let $\pi_k=p\bigl((k+\tfrac12)/(N+1)\bigr)$, $k=0,\dots,N$, be the values of
the coefficient at the midpoints of the grid, and let $A$ be the symmetric
irreducible tridiagonal matrix with entries $a_i=\pi_i+\pi_{i+1}$ and
$b_i=-\pi_{i+1}$, the three-point discretization of $-(p(x)w')'$ on $(0,1)$
with Dirichlet conditions and $p>0$. The matrix to be encoded is $S=A^{-1}$,
which is a one-pair matrix by  Theorem~\ref{thm:inv-tridiag}, and
writing $q_k=1/\pi_k$ and $\Sigma=\sum_{k=0}^{N}q_k$ its generators are
\[
u_i=\frac{1}{\sqrt\Sigma}\sum_{k\le i}q_k,
\qquad
v_j=\frac{1}{\sqrt\Sigma}\sum_{k>j}q_k,
\qquad i,j=0,\dots,N-1 ,
\]
or equivalently $Au=e_{N-1}/v_{N-1}$ and $Av=e_0/u_0$.
They are the discrete counterparts of $U(x)=\int_0^x dt/p(t)$ and $V(x)=\int_x^1 dt/p(t)$, the two solutions of $-(pw')'=0$ meeting one of the two boundary conditions. Assume now $p$ analytic on $[0,1]$ and bounded away from zero. Then $1/p$ is
analytic, so $U$ and $V$ are analytic and the degree required by the
construction above is $d_\varepsilon=O(\log(1/\varepsilon))$, independent of
$N$, which is the first condition. Moreover $U$ and $V$ are monotone and free
of interior peaks, so their values are of comparable size across the grid,
which is the second condition. The routine above prepares samples of $U$ and
$V$, while the exact generators are the cumulative sums written above. The two
vectors are not equal: the sum $\sum_{k\le i}q_k$ is the midpoint quadrature
rule for the integral $(N+1)\int_0^{x_i}dt/p(t)$, so the relative difference
between the two is $O(N^{-2})$. This difference adds to the error $\varepsilon$
of hypothesis~2 of Theorem~\ref{thm:main}.

\subsection{Assembly and proof of  Theorem~\ref{thm:main}}\label{sec:assembly}

\begin{proof}
We first describe the construction assuming the generators $u,v$ are loaded
exactly, and show that it produces an exact block-encoding of $S(u,v)/(2\sqrt N)$.
We then account for the state-preparation error, which is the only source of inexactness. Furthermore, Since only the first $N$ columns of the realized unitary are needed, we project every operator onto $\ket{0}^{2n+2}\otimes I_N$.

\noindent The circuit used to construct the block-encoding can be organized into five
specific registers:
\begin{itemize}
    \item An $n$-qubit register used to load the generator $v$;
    \item A one qubit register, that we  call the triangular register, used to index the lower or upper triangular part of the matrix;
    \item An $n$-qubit register used to load the generator $u$;
    \item A one qubit register, that we  call the  placeholder register,   used for doubling the space of action of the shifts, allowing us to move some spurious elements and replacing them with zeros, when constructing the lower and upper triangular matrices;
    \item A $n$ qubits register, that we  call the replication register,  used to replicate the generators.
\end{itemize}
In total the circuit acts on $3n+2$ qubits.

We now review the main steps of the algorithm. We  refer to Figure~\ref{fig:algo} for a pictorial representation of the algorithm and to Figure~\ref{fig:circ} for the complete circuit.

\paragraph{State preparation}
The first step is to load into the quantum state the generators $u,\,v\in \mathbb{R}^N$.
Assuming an efficient state-preparation circuit, we define the two operators $U_u$ and $U_v$,
which encode the two normalized vectors into the quantum state, i.e. $U_u \ket{0}= u, U_v\ket{0}=v$.
We define $ U_{\rm gen} \in \mathbb{R}^{4N^3\times 4N^3}$, $U_{\rm gen}=U_v\otimes I_2 \otimes U_u \otimes I_{2}\otimes I_{N}$ as the operator that loads both generators in the quantum circuit. Looking a the first $N$ columns of the operator, that is projecting it onto $\ket{0^{2n+2}}\otimes I_N$
$$
U_{\mathrm{gen}}(\ket{0^{2n+2}}\otimes I_N)=\begin{bmatrix}
    v_0 u\\
    0_{N}\\
   v_1 u\\
   0_{N}\\
    \vdots\\
     v_{N-1} u\\\
    0_{N}\\
\end{bmatrix}\otimes (\ket {0} \otimes I_N),
$$
where each vector $v_i u$ has length $N$ and because of the external Kroneker product with $I_{N}$, the values $v_i u_j$ and $v_i u_{j+1}$ are separated by $N$ zeros.

\paragraph{Reordering}
The swap operator $U_{\mathrm{SWAP}}$ exchanges the $v$-register with the
replication register; it moves the column-indexing identity $I_N$ to the leading
register and collects the loaded amplitudes into the row blocks
$u_j v$, for $j=0,\dots,N-1$. We denote the result by
$U_1=U_{\mathrm{SWAP}}\,U_{\mathrm{gen}}$, we have
$$
U_1(\ket{0^{2n+2}}\otimes I_N)=(I_N \otimes \ket{0}) \otimes \begin{bmatrix}
    u_0 v\\
    0_{N}\\
   u_1 v\\
   0_{N}\\
    \vdots\\
     u_{N-1} v\\\
    0_{N}\\
\end{bmatrix}.
$$

\paragraph{Duplication and the second branch}
A Hadamard gate on the triangular qubit duplicates each block, $U_2=(I_N\otimes H\otimes I_{2N^{2}})\,U_1$, producing, up to the global factor
$1/\sqrt2$, two identical branches: one will form the lower-triangular part $L$
of $S$, the other its strict upper-triangular part $R$. Since $S$ is symmetric,
the strict-upper entries are obtained from the lower ones by reversing the order
of the Kronecker product; applying the controlled-\textsc{swap} operator
$U_{\text{C-SWAP}}$ of Section~\ref{sec:cswap} to the second branch turns $u_j v$ into
$v_j u$, yielding $U_3=U_{\text{C-SWAP}}\,U_2$. Writing
$\Lt_j=u_j v$ and $\Rt_j=v_j u$ for the two (still unprocessed)
column blocks,
$$
  U_3(\ket{0}_{2n+2} \otimes I_N)=I_N\otimes\tfrac{1}{\sqrt2}
  \begin{bmatrix}
    \Lt_0\\ 0_N\\ \vdots\\ \Lt_{N-1}\\ 0_N\\ \hline
    \Rt_0\\ 0_N\\ \vdots\\ \Rt_{N-1}\\ 0_N
  \end{bmatrix} ,
$$
so that each padded block $[\Lt_j;0_N]$ and $[\Rt_j;0_N]$ has length $2N$ and still
carries spurious entries, removed in the next steps.

\paragraph{Selecting the triangular entries}
The operator $Z_{\mathrm{sel}}$ of Section~\ref{sec:zsel} acts within each $2N$-block, applying $Z_{2N}^{\,j}$ to the
$j$-th lower block and $(Z_{2N}^{\dagger})^{N-j}$ to the $j$-th upper block. This
brings the wanted entries of $L$ and $R$ into the leading half of each block and
routes the spurious entries into the auxiliary (trailing) half, giving
$U_4=Z_{\mathrm{sel}}\,U_3$:
\[
  U_4(\ket{0}_{2n+2} \otimes I_N)=I_N\otimes\tfrac{1}{\sqrt2}
  \begin{bmatrix}
    Z_{2N}^{0}\,[\Lt_0;0]\\ Z_{2N}^{1}\,[\Lt_1;0]\\ \vdots\\
    Z_{2N}^{N-1}\,[\Lt_{N-1};0]\\
    (Z_{2N}^{\dagger})^{N}\,[\Rt_0;0]\\ \vdots\\
    (Z_{2N}^{\dagger})^{1}\,[\Rt_{N-1};0]
  \end{bmatrix}.
\]
After this step the leading halves contain exactly the entries of $L$ and $R$,
though not yet in their final rows.

In particular we have
\begin{equation} \label{eq:afterZgup}
Z_{2N}^{k-1} \begin{bmatrix}
    \tilde L_{k} \\ 0_N
\end{bmatrix}= \begin{bmatrix}
   \tilde L_{k}(k:N) \\ 0_{k-1} \\  \hline  0_{N-k+1} \\ \tilde L_{k}({1:k-1}) 
\end{bmatrix}, \quad 
\end{equation}
where $L_{k}(k:N)$ contains  the  entries of the $k$-th column of $L$, and 
\begin{equation}
\label{eq:afterZgdown}
(Z_{2N}^\dagger)^k \begin{bmatrix}
    \tilde R_{N-k+1}  \\ 0_N
\end{bmatrix}= \begin{bmatrix}
   0_k \\ \tilde R_{N-k+1} (1:N-k) \\ \hline 
   
   \tilde R_{N-k+1} (N-k+1:N)\\ 0_{N-k}
\end{bmatrix}.
\end{equation}

At this point, the two sets of blocks contain exactly the correct entries
of $L$ and $R$, but not yet in the positions needed to assemble the columns
of $S$.

\paragraph{Placing the entries}
The operator $Z_{\mathrm{fix}}$ of Section~\ref{sec:zfix} completes the placement, shifting the leading $N$ entries of
each lower block down by $k$ and those of each upper block up by $N-k$, so that
every entry reaches its final row. The resulting state $U_5=Z_{\mathrm{fix}}\,U_4$ carries the
same $1/\sqrt2$ normalization as $U_4$ and has the columns of $L$ and $R$
correctly positioned, with the surplus amplitudes gathered in the trailing $N$
entries (denoted $\star$) of each $2N$-block.

A second Hadamard gate on the triangular qubit interferes the two branches and
assembles each column $S_i=L_i+R_i$, giving
$U_6=(I_N\otimes H\otimes I_{2N^{2}})\,U_5$:
\[
  U_6(\ket{0}_{2n+2} \otimes I_N)=I_N\otimes\tfrac{1}{2}
  \begin{bmatrix}
    S_0 \\ \star\\ S_1\\ \star\\ \vdots\\ S_{N-1}\\ \star\\ \hline\circledast
  \end{bmatrix},
\]
where $S_i$ is the $i$-th column of $S$, each $\star\in\R^{N}$ is a garbage block,
and $\circledast\in\R^{2N^{2}}$ collects the remaining garbage.

\paragraph{Final mixing}
At this stage the tensor factor $I_N \otimes$ means that the stacked columns
$S_0,\dots,S_{N-1}$ are replicated identically (but shifted of $N$ positions down respect to the previous column) across the $N$ copies indexed by
the replication register. The permutation operator $Z_{\mathrm{flip}}$ of Section~\ref{sec:zflip} selects, within the
$j$-th copy, the column $S_j$ and moves it to the head of that column-block, producing
$U_7=Z_{\mathrm{flip}}\,U_6$, in which the $j$-th column-block carries $S_j$ at its top.
Applying $n$ Hadamard gates to the leading register superposes the $N$ copies,
$U_{\mathrm{BE}}=(H^{\otimes n}\otimes I_{4N^{2}})\,U_7$, and places the columns
side by side:
\[
  U_{\mathrm{BE}}(\ket{0}_{2n+2} \otimes I_N)=\frac{1}{2\sqrt N}
  \begin{bmatrix}
    S_0 & S_1 & \cdots & S_{N-1}\\
    \star & \star & \cdots & \star\\
    \vdots & \vdots & \ddots & \vdots\\
    \star & \star & \cdots & \star
  \end{bmatrix}.
\]
The top-left $N\times N$ block equals $S(u,v)/(2\sqrt N)$. All building blocks are
Hadamard gates, swaps, and shift operators, each of $O(\polylog N)$ depth as stated in Remark~\ref{rem:cost};
together with the assumed $O(\polylog N)$ state preparation, the whole circuit
has depth $O(\polylog N)$. Hence, when $u$ and $v$ are loaded exactly,
$U_{\mathrm{BE}}$ is an exact $(2\sqrt N,\,2n+2)$-block-encoding of $S(u,v)$.

In general the state preparation routines, load unit-norm approximations $\tilde u,\tilde v$, so the 
circuit described above realizes a block-encoding of $S(\tilde u,\tilde v)/(2\sqrt N)$, from now on we set $S(\tilde u, \tilde v)=\tilde S$.
By Definition~\ref{def:be} with $\alpha=2\sqrt N$, the block-encoding error is
$\norm{S(u,v)-S(\tilde u,\tilde v)}_2$, which we now bound. Set
$\delta_u=u-\tilde u$ and $\delta_v=v-\tilde v$. The map $(u,v)\mapsto S(u,v)$ is
bilinear, so by linearity in each argument
$$
  S(u,v)-S(\tilde u,\tilde v)=S(\delta_u,v)+S(\tilde u,\delta_v).
$$
For any $a,b\in\R^{N}$ the Frobenius norm of $S(a,b)$ satisfies
$$
  \norm{S(a,b)}_F^{2}
  =2\sum_{i<j}a_i^{2}b_j^{2}+\sum_i a_i^{2}b_i^{2}
  \le 2\,\norm{a}_2^{2}\,\norm{b}_2^{2}.
$$
Since $\norm{\cdot}_2\le\norm{\cdot}_F$, applying this bound to the two terms gives
$$
  \norm{S(u,v)-S(\tilde u,\tilde v)}_2
  \le\sqrt2\bigl(\norm{v}_2\norm{\delta_u}_2+\norm{\tilde u}_2\norm{\delta_v}_2\bigr).
$$
Using the unit-norm assumptions
$\norm{v}_2=\norm{\tilde u}_2=1$ and $\norm{\delta_u}_2,\norm{\delta_v}_2\le\varepsilon$,
we conclude $$ \norm{S(u,v)-S(\tilde u,\tilde v)}_2\le 2\sqrt2\,\varepsilon,$$
so the circuit implements a
$\bigl(2\sqrt N,\,2n+2,\,2\sqrt2\,\varepsilon\bigr)$-block-encoding of $S(u,v)$.
In particular, for $\varepsilon=0$ the encoding is exact. 
\end{proof}

\section{The unsymmetric case}
\label{sec:unsymm}

The construction extends to general (non-sym\-me\-tric) generator-representable
semi\-se\-pa\-rable matrices at the sole cost of making the state preparation
controlled and of replacing the Hadamard gate on the triangular qubit by a
suitable rotation. Before stating the result we discuss why the unit-norm
normalization of Theorem~\ref{thm:main} cannot be carried over unchanged.

A one-pair matrix $S(u,v,x,y)$ of Definition~\ref{def:onepair-gen_unsymm} is
determined by four generators subject to the diagonal compatibility
$u_iv_i=x_iy_i$, $i=0,\dots,N-1$, which makes the two triangular expressions
agree on the diagonal. Requiring in addition that all four generators be
unit-norm is an overdetermined constraint: normalizing $u,v,x,y$ independently
rescales the two triangular parts by the different factors $\|u\|_2\|v\|_2$ and
$\|x\|_2\|y\|_2$, and the compatibility is destroyed unless these two products
happen to coincide. In other words, one cannot in general prescribe both
$\|u\|_2=\|v\|_2=\|x\|_2=\|y\|_2=1$ and $u_iv_i=x_iy_i$.

We therefore drop the normalization hypothesis altogether and block-encode
$S=S(u,v,x,y)$ for arbitrary generators. This is possible because the two
triangular parts are separately bilinear in their own pair of generators.
Writing $L(u,v)$ for the lower triangular part including the diagonal,
$L_{ij}=u_jv_i$ for $i\ge j$, and $R(x,y)$ for the strictly upper one,
$R_{ij}=x_iy_j$ for $i<j$, we have $S=L(u,v)+R(x,y)$ and, denoting by
$u_n=u/\|u\|_2$ the normalized generators and setting
$$
  w_{\rm lo}=\|u\|_2\|v\|_2,\qquad w_{\rm up}=\|x\|_2\|y\|_2,
$$
bilinearity gives
\begin{equation}\label{eq:split}
  S(u,v,x,y)=w_{\rm lo}\,L(u_n,v_n)+w_{\rm up}\,R(x_n,y_n).
\end{equation}
The circuit thus loads the four normalized generators, which is what any
state-preparation routine produces, while the two scales
$w_{\rm lo},w_{\rm up}$ are classically computable numbers that are
absorbed into the subnormalization. The diagonal is supplied by the lower
branch and equals $w_{\rm lo}\,u_{n,i}v_{n,i}=u_iv_i$, i.e.\ the correct
entry: the compatibility condition is imposed on the original generators, as in
Definition~\ref{def:onepair-gen_unsymm}, and no compatibility is required of
their normalized counterparts.

\begin{theorem}\label{thm:unsym}
Let $N=2^n$ and let $S=S(u,v,x,y)\in\R^{N\times N}$ be a one-pair matrix as in
Definition~\ref{def:onepair-gen_unsymm}, with arbitrary nonzero generators.
Assume that
\begin{enumerate}
\item the state-preparation routines load approximations
      $\tilde u_n,\tilde v_n,\tilde x_n,\tilde y_n$ of the normalized
      generators, with error at most $\varepsilon$ in the Euclidean norm;
\item these routines have circuit depth $O(\mathrm{polylog}\,N)$.
\end{enumerate}
Then there exists a quantum circuit acting on $3n+2$ qubits, with depth
$O(\mathrm{polylog}\,N)$, that implements an
$\bigl(\alpha,\,2n+2,\,2\sqrt{w_{\rm lo}^2+w_{\rm up}^2}\,
\varepsilon\bigr)$-block-encoding of $S$, where
\begin{equation}\label{eq:alpha-unsym}
  \alpha=\sqrt{N}\,\bigl(w_{\rm lo}+w_{\rm up}\bigr)
        =\sqrt{N}\,\bigl(\|u\|_2\|v\|_2+\|x\|_2\|y\|_2\bigr).
\end{equation}
\end{theorem}

\begin{proof}[Proof sketch]
The circuit is that of Theorem~\ref{thm:main} with two modifications.

The triangular qubit, which already selects the lower versus
the strict upper branch, is also used to select the pair of generators: on
$\ket 0$ the pair $(u_n,v_n)$ is loaded, on $\ket 1$ the pair $(x_n,y_n)$, into
the same registers. Formally $U_{\rm gen}$ is replaced by a multiplexed
operator whose depth is the sum of the two state preparations, hence still
$O(\mathrm{polylog}\,N)$; no additional qubit is required, so the circuit still
acts on $3n+2$ qubits.

 The two Hadamard gates acting on the
triangular qubit in the symmetric case, are now replaced by the rotation $R_y(\vartheta)$ with
\begin{equation}\label{eq:theta}
  \vartheta=2\arctan\sqrt{\frac{w_{\rm up}}{w_{\rm lo}}},
\end{equation}
so that the branches are created, and later recombined, with amplitudes
$\cos(\vartheta/2)=\sqrt{w_{\rm lo}/(w_{\rm lo}+w_{\rm up})}$ and
$\sin(\vartheta/2)=\sqrt{w_{\rm up}/(w_{\rm lo}+w_{\rm up})}$.
Using the same angle in both gates makes the amplitude collected by each branch
proportional to the square of its weight, hence to its scale, so that the two
terms of \eqref{eq:split} are recombined in the correct proportion with the
single common factor $(w_{\rm lo}+w_{\rm up})^{-1}$. 

The remaining steps are unchanged. After $U_{\rm SWAP}$ and $U_{\text{C-SWAP}}$
the $\ket 0$ branch carries $\tilde L_j=u_{n,j}v_n$ and the $\ket 1$ branch
$\tilde R_j=y_{n,j}x_n$; applying $Z_{\rm sel}$ and $Z_{\rm fix}$ as before, the
lower branch retains the entries with $i\ge j$ and the upper one those with
$i<j$. The final rotation on the triangular qubit adds the two branches, and
$Z_{\rm flip}$ together with $H^{\otimes n}$ place the columns side by side,
contributing the factor $\sqrt N$. Collecting the factors gives
$S/\alpha$ in the leading block with $\alpha$ as in \eqref{eq:alpha-unsym}.

 The circuit encodes
$$\tilde S=w_{\rm lo}L(\tilde u_n,\tilde v_n)
+w_{\rm up}R(\tilde x_n,\tilde y_n),$$  and 
$w_{\rm lo}$ and $w_{\rm up}$ are classically computed values so they carry no error.
The two terms act on disjoint index sets, so
$\|S-\tilde S\|_F^2=w_{\rm lo}^2\|L-\tilde L\|_F^2
+w_{\rm up}^2\|R-\tilde R\|_F^2$. By bilinearity and the bound
$\|L(a,b)\|_F^2=\sum_{i\ge j}a_j^2b_i^2\le\|a\|_2^2\|b\|_2^2$ one gets
$\|L-\tilde L\|_F\le 2\varepsilon$ and likewise $\|R-\tilde R\|_F\le
2\varepsilon$, whence
$\|S-\tilde S\|_2\le\|S-\tilde S\|_F\le
2\sqrt{w_{\rm lo}^2+w_{\rm up}^2}\,\varepsilon$.
\end{proof}

\begin{rem}
Theorem~\ref{thm:main} is recovered as the special case $x=u$, $y=v$. Then
$w_{\rm lo}=w_{\rm up}=\|u\|_2\|v\|_2$, the angle \eqref{eq:theta}
becomes $\vartheta=\pi/2$ and $R_y(\pi/2)$ reduces to the Hadamard gate of the
symmetric construction, while \eqref{eq:alpha-unsym} gives
$\alpha=2\sqrt{N}\,\|u\|_2\|v\|_2$ and the error bound becomes
$2\sqrt2\,\|u\|_2\|v\|_2\,\varepsilon$. In particular, for unit-norm generators
one obtains the $(2\sqrt N,2n+2,2\sqrt2\,\varepsilon)$-block-encoding of
Theorem~\ref{thm:main}; more generally, the same argument shows that the
unit-norm assumption can be dropped there as well, at the price of the explicit
factor $\|u\|_2\|v\|_2$ in the subnormalization.
\end{rem}

\section{Comparison with other encoding}\label{sec:comparison}
In this section we compare our construction with existing block-encoding
schemes that are applicable to generator-repre\-sen\-ta\-ble semiseparable
matrices, either because they are tailored to (a superclass of) this
structure, or because they apply to arbitrary dense matrices. The relevant
parameters for the comparison are the subnormalization factor $\alpha$, the
number of ancilla qubits, the error of the final encoding,  the circuit depth or gate count, and the data-access model, that is, the form in which the matrix
must be supplied to the circuit: state-preparation routines for the
generators, oracle access to matrix entries, or classical data compiled
directly into the circuit. 

The scheme of~\cite{qseparxiv} is, to the best of our knowledge, the
only previous block encoding tailored to one-pair matrices, and is therefore
the closest term of comparison. It is based on a triangular-diagonal
factorization of $S$, and while it achieves polylogarithmic depth, the
factorization is paid for by a subnormalization factor growing as $O(N^2)$,
together with $2\log N+7$ ancillas. Our construction bypasses any
factorization and works with the generators alone: the only matrix-dependent
ingredients are the two state-preparation routines, no classical
pre-processing of the matrix is required, and the subnormalization drops
from $O(N^2)$ to $2\sqrt{N}$. The improvement in $\alpha$ is not a detail of
the encoding but propagates through any subsequent algorithm: in quantum
singular value transformation and related solvers, the cost of applying a
function of the encoded matrix scales with $\alpha$ (e.g., as
$O(\alpha\kappa\,\mathrm{polylog}(1/\epsilon))$ for linear systems with
condition number $\kappa$~\cite{gilyen2019quantum,Chakraborty19}), so the reduction from $N^2$ to $\sqrt{N}$ yields
a polynomial saving on the whole pipeline.

We remark that $\alpha=2\sqrt{N}$ is not optimal in absolute terms: any
$(\alpha,m,\varepsilon)$-block-encoding requires $\alpha\ge\|S\|_2$, and for
unit-norm generators $\|S\|_2\le\|S\|_F\le\sqrt{2}$, so the overhead
$\alpha/\|S\|_2$ can be as large as $\Theta(\sqrt{N})$. The factor
$\sqrt{N}$ originates from the final Hadamard layer that superposes the $N$
column copies, and is analogous to the normalization loss incurred by
FABLE-type constructions. Whether it can be removed within a
generator-based, factorization-free construction is an open question. In Figure~\ref{fig:norm} we plot the ratio $r(N)=\alpha/\|S\|_2$ for semiseparable matrices whose generators are constructed in three different ways (see Section~\ref{sec:experiments} for the details). For $\alpha=2\sqrt{N}$ we see the overhead for our method while for $\alpha= N$ we see the overhead of FABLE.

The comparison in terms of the error in the encoding shows another improvement of the algorithm of this paper over~\cite{qseparxiv}. The error in~\cite{qseparxiv} grows as $O(N^2 \varepsilon),$ where here is $2\sqrt{2} \varepsilon,$ independent on the size of the encoded matrix.

Semiseparable matrices are the simplest instance of the hierarchical
rank-struc\-tu\-red classes addressed in \cite{NKL22,TL26}, so the natural term of comparison is the block encoding of general $\mathcal H$-matrices of \cite[Section~4.3]{NKL22}. In this framework, $S(u,v)$ is
split by recursive bisection of the index set into $O(N)$ off-diagonal
blocks distributed over $\log N$ levels; an interesting fact is that every such block in the semiseparable case is exactly of rank
one, being the outer product of 
 two contiguous portions of the generators. The scheme of
\cite{NKL22} block-encodes each block from state preparations of the
normalized restrictions of generators, assembles each level as a block-sparse matrix, and
sums the levels by a linear combination of unitaries. For unit-norm
generators the resulting subnormalization is $O(\log N)$, asymptotically
smaller than our $2\sqrt N$; the gain, however, is paid in the access
model, which requires state-preparation oracles multiplexed over all blocks
of the hierarchy together with an oracle for the $O(N)$ partial norms of the portions of the generators, i.e., classically precomputed data of the
same size as the generators themselves. Moreover, the input inaccuracy of
these oracles is amplified by the subnormalization and accumulated across
the levels. Our construction is different;
it requires nothing beyond two plain state preparations of the full
generators, involves no classical pre-processing, no multiplexing, and no
propagation of errors, at the price of a larger subnormalization. The direct scheme of \cite{TL26}
for hierarchically block separable matrices lowers the oracle access to
$O(N)$ entries, but builds the encoding recursively by multiplying
per-level block encodings, so that both the subnormalization and the error
propagate multiplicatively along the $O(\log N)$ levels of the recursion,
yielding bounds that grow polynomially with $N$ on this class; their
alternative sparsification approach achieves constant subnormalization but
encodes an extended sparse system rather than $S$ itself, and is a tool for
linear solves rather than a block encoding of the matrix.

\begin{figure}[htbp]
\centering
\begin{minipage}[t]{0.48\textwidth}
\vspace{-8cm}
  \centering
  \begin{tikzpicture}
  \begin{axis}[
      width=\linewidth,
      height=4.5cm,
      xlabel={$N$},
      ylabel={$\alpha / \|S\|_2$},
      xlabel near ticks,
      ylabel near ticks,
      xlabel style={font=\tiny},
      ylabel style={font=\tiny},
      title style={font=\tiny, yshift=-2pt},
      tick label style={font=\tiny},
      grid=both,
      grid style={dotted},
      minor tick num=0,
      enlarge x limits=false,
      legend style={font=\tiny, draw=none, fill=none, at={(0.02,0.98)}, anchor=north west},
  ]
  \addplot[cyan!60!black,thin, mark=*, mark size=2pt] coordinates {
  (4,4.0522) (8,6.4346) (16,10.479) (32,17.787) (64,26.382) (128,27.681) (256,51.278) (512,65.303) (1024,99.914)
  }; \addlegendentry{Random}
  \addplot[blue, thin, mark=o] coordinates {
  (4,4.0522) (8,9.0999) (16,20.959) (32,50.308) (64,105.53) (128,156.59) (256,410.22) (512,738.81) (1024,1598.6)
  }; \addlegendentry{FABLE}
  \end{axis}
  \end{tikzpicture}

  \vspace{0.3cm}

  \begin{tikzpicture}
  \begin{axis}[
      width=\linewidth,
      height=4.5cm,
      xlabel={$N$},
      ylabel={$\alpha / \|S\|_2$},
      xlabel near ticks,
      ylabel near ticks,
      xlabel style={font=\tiny},
      ylabel style={font=\tiny},
      title style={font=\tiny, yshift=-2pt},
      tick label style={font=\tiny},
      grid=both,
      grid style={dotted},
      minor tick num=0,
      enlarge x limits=false,
      legend style={font=\tiny, draw=none, fill=none, at={(0.02,0.98)}, anchor=north west},
  ]
  \addplot[red, thin, mark=square*, mark size=2pt] coordinates {
  (4,3.887) (8,5.6561) (16,8) (32,11.314) (64,16) (128,22.627) (256,32) (512,45.255) (1024,64)
  }; \addlegendentry{Tridiagonal} 
  \addplot[blue, thin, mark=o] coordinates {
  (4,3.887) (8,7.999) (16,16) (32,32) (64,64) (128,128) (256,256) (512,512) (1024,1024)
  }; \addlegendentry{FABLE}
  \end{axis}
  \end{tikzpicture}
\end{minipage}
\hfill
\begin{minipage}[t]{0.48\textwidth}
  \centering
  \begin{tikzpicture}
  \begin{axis}[
      width=\linewidth,
      height=8.7cm,
      xlabel={$N$},
      ylabel={$\alpha / \|S\|_2$},
      xlabel near ticks,
      ylabel near ticks,
      xlabel style={font=\tiny},
      ylabel style={font=\tiny},
      title style={font=\tiny, yshift=-2pt},
      tick label style={font=\tiny},
      grid=both,
      grid style={dotted},
      minor tick num=0,
      enlarge x limits=false,
      legend style={font=\tiny, draw=none, fill=none, at={(0.02,0.98)}, anchor=north west},
  ]
  \addplot[green!60!black, thin, mark=triangle*, mark size=2pt] coordinates {
  (4,24.266) (8,38.917) (16,56.845) (32,81.047) (64,114.85) (128,162.51) (256,229.85) (512,325.07) (1024,459.72)
  }; \addlegendentry{Exponential}
  \addplot[blue, thin, mark=o] coordinates {
  (4,24.266) (8,55.038) (16,113.69) (32,229.24) (64,459.41) (128,919.28) (256,1838.8) (512,3677.7) (1024,7355.5)
  }; \addlegendentry{FABLE}
  \end{axis}
  \end{tikzpicture}
\end{minipage}
\caption{Comparison of the subnormalization ratio $r(N)=\alpha/\|S\|_2$ between our method and FABLE, for different generator-pair types (one per plot).}
\label{fig:norm}
\end{figure}
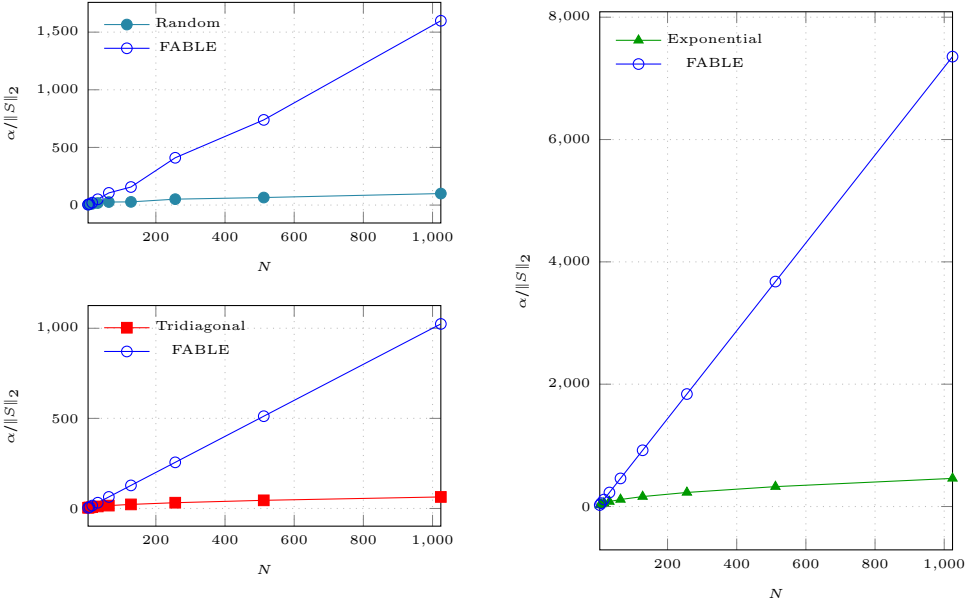

FABLE \cite{fable} block-encodes arbitrary dense matrices by compiling
the classical entries into $O(N^2)$ controlled rotations, with $n+1$
ancillas and subnormalization $N$ for entries bounded by one; the circuit is
often compressible, but no compression guarantee holds in general. Since
FABLE assumes no structure, the comparison measures the value of structure:
on generator-representable semiseparable matrices our circuit replaces the
$O(N^2)$ gate count with polylogarithmic depth and improves the
subnormalization from $N$ to $2\sqrt{N}$.
\paragraph{Obtaining a block encoding of the inverse}
The construction can also be compared with the route that obtains $A^{-1}$ by
inverting a block encoding of $A$. We take as test case the family of
Section~\ref{sec:stateprep}, where $A$ is the three-point discretization of a
Sturm-Liouville operator, $S=A^{-1}$ is a one-pair matrix and the generators
are explicit. The construction of Theorem~\ref{thm:main} returns a block
encoding of $A^{-1}$ with subnormalization $2\sqrt N\|u\|_2\|v\|_2$ at the cost
of a single circuit of polylogarithmic depth. Quantum singular value
transformation applied to a block encoding of $A$ returns instead a block
encoding of $A^{-1}$ whose subnormalization $\Theta(\|A^{-1}\|_2)$ is optimal
up to a constant factor, since the polynomial approximating $1/x$ on the
spectral interval is bounded by one in modulus, but it requires that polynomial
to have degree $\widetilde O(\kappa)$, hence that many queries to the encoding
of $A$, each of polylogarithmic depth because $A$ is
tridiagonal~\cite{camps2024explicit}. Our approach is cheaper in depth by a factor $\kappa$ and more
expensive in subnormalization by a factor $\Theta(\sqrt N)$, because
$\|u\|_2\|v\|_2=\Theta(\|A^{-1}\|_2)$ on this family, and for the
discretization of a second-order operator with $p$ analytic and bounded away
from zero on $[0,1]$ one has $\kappa=\Theta(N^{2})$.
If we apply amplitude amplification to both  algorithms, our approach wins by a factor $\widetilde O(N^{3/2})$.

\section{Numerical experiments} \label{sec:experiments}
The experiments we perform to validate our quantum algorithm are aimed in three
different directions:
\begin{enumerate}
\item correctness of the circuit and of its QCLAB~\cite{qclab} implementation;
\item tightness of the error bound;
\item comparison with FABLE~\cite{fable} in terms of circuit depth and size. 
\end{enumerate}
We run our tests on three different classes of generator-representable matrices.
The first class, {\sc Random}, is obtained by constructing the semiseparable
matrix $S(u,v)$ where $u$ and $v$ are normalized random vectors with entries
drawn independently from the standard normal distribution (zero mean, unit
variance). The second family, which we call {\sc Tridiagonal}, consists of
generators extracted from the semiseparable inverse of a random symmetric
tridiagonal matrix, constrained to be irreducible and diagonally dominant. The
generators are obtained in a numerically stable way from the first and last
columns of the inverse, which are proportional to $v$ and $u$ respectively, each obtained by solving a single tridiagonal linear system; the semiseparable matrix is then
reconstructed from the normalized generators to enforce the unit-norm
assumption. The third class, denoted {\sc Exponential}, is constructed as
follows: for $i=0,\ldots,N-1$ we set $u_i=e^{\theta t_i}$ and $v_i=e^{-\theta
t_i}$, where $t_i=i/N$ and, in our tests, $\theta=3$. Once normalized, these
generators have two useful features: they are the generators of the covariance
matrix of a stationary Ornstein--Uhlenbeck (Gauss--Markov) process sampled on a
uniform grid, and, being samples of a smooth exponential, they can be prepared
as quantum states by a circuit of polylogarithmic depth~\cite{li2023efficient}. In this particular case, since $e^{\theta i/N}=\prod_{k}\bigl(e^{\theta 2^{k}/N}\bigr)^{b_k}$
   factorizes over the binary digits $b_k$ of $i$, the normalized state is a
   product state that can be prepared exactly by $n=\log_2 N$ single-qubit $R_y$
   rotations, with no entangling gates.

The code reproducing all experiments is
available both in QCLAB~\cite{qclab} and  Qiskit~\cite{qiskit2024}
at~\url{https://github.com/Giacomo-Antonioli/Semiseparable_Block_Encoding}.
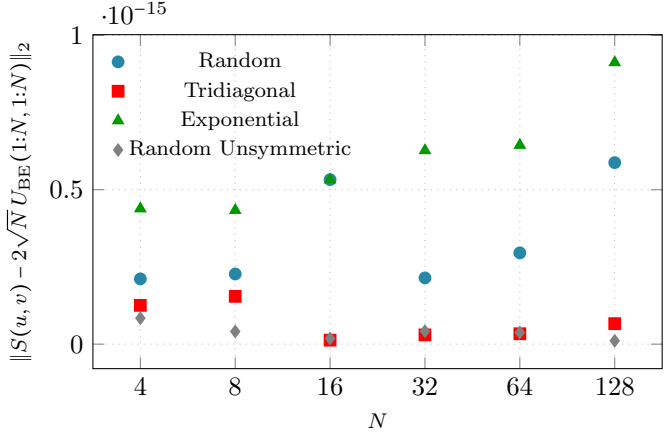
\begin{figure}[htbp]
  \centering
  \begin{tikzpicture}
  \begin{axis}[
      width=0.7\linewidth,
      height=6.0cm,
      xmode=log,
      log basis x=2,
      xtick={4,8,16,32,64,128},
      xticklabels={4,8,16,32,64,128},
      xlabel={$N$},
      ylabel={ $\|S(u,v)-2\sqrt{N}\,U_{\mathrm{BE}}(1{:}N,1{:}N)\|_2$},
      grid=both,
      xlabel style={font=\footnotesize},
      ylabel style={font=\footnotesize},
      grid style={dotted},
      minor tick num=0,
      enlarge x limits=true, 
      legend style={
          font=\footnotesize,
          draw=none,
          fill=none,
          at={(0.02,0.98)}, 
          anchor=north west
      }
  ]

  \addplot[cyan!60!black, thick, mark=*, only marks] coordinates {
  (4, 2.113294e-16) (8, 2.271675e-16) (16, 5.327904e-16) (32, 2.145947e-16) (64, 2.956085e-16) (128, 5.878878e-16) 
  };
  \addlegendentry{Random}

  \addplot[red, thick, mark=square*, only marks] coordinates {
  (4, 1.255208e-16) (8, 1.546673e-16) (16, 1.293884e-17) (32, 2.999050e-17) (64, 3.361439e-17) (128, 6.639853e-17) 
  };
  \addlegendentry{Tridiagonal}

  \addplot[green!60!black, thick, mark=triangle*, only marks] coordinates {
  (4, 4.387320e-16) (8, 4.336215e-16) (16, 5.322507e-16) (32, 6.272112e-16) (64, 6.444265e-16) (128, 9.121271e-16) 
  };
  \addlegendentry{Exponential}
 \addplot[gray, thick, mark=diamond*, only marks] coordinates {
 (4, 8.407295e-17) (8, 4.092751e-17) (16, 1.827259e-17) (32, 4.189104e-17) (64, 3.770931e-17) (128, 1.077174e-17) 
  };
  \addlegendentry{Random Unsymmetric}

  \end{axis}
  \end{tikzpicture}
  \caption{Correctness of the encoding. For different testing models and different matrix size we report the error of the encoding assuming the generators correct up to machine precision. We see that the absolute error is $O(10^{-15})$, comparable with the expected  theoretical error.}
    \label{fig:correctness}
\end{figure}
\subsection{Correctness}
For semiseparable matrices of the three families we report the block-encoding
error $\|S-\tilde S\|_2$, where $\tilde S=2\sqrt{N}\,U_{\mathrm{BE}}(1{:}N,1{:}N)$
is the reconstructed leading block. As shown in Figure~\ref{fig:correctness}, the error is of order $10^{-15}$ for all
three families and all tested sizes, in agreement with the roundoff expected from
computing the Euclidean norm of the difference between the exact and the encoded
matrix; this confirms that the circuit realizes an exact block encoding when the
generators are loaded exactly.

\subsection{Tightness of the error bound}
To assess the tightness of the bound of Theorem~\ref{thm:main} we perturb the
generators by an error of order $\varepsilon$ in each component. Once the
unit-norm generators $u,v$ of each model are constructed, we draw random directions
$\delta_u,\delta_v=\mathrm{randn}(N,1)$ and set $\tilde u = u+\varepsilon\,\delta_u/\|\delta_u\|_2$,
$\tilde v = v+\varepsilon\,\delta_v/\|\delta_v\|_2$, which we then normalize to unit norm.
The resulting approximations satisfy $\|u-\tilde u\|_2,\|v-\tilde v\|_2 =
\varepsilon + O(\varepsilon^2)$, so that hypotheses~1 and~2 of
Theorem~\ref{thm:main} hold up to higher-order terms. We then apply the circuit
and measure $\|S(u,v)-2\sqrt{N}\,U_{\mathrm{BE}}(1{:}N,1{:}N)\|_2$.
\begin{figure}[htbp]
  \centering
  \begin{subfigure}[t]{0.48\linewidth}
    \centering
    \begin{tikzpicture}
    \begin{axis}[
        width=0.98\linewidth,
        height=5.0cm,
        grid=both,
        grid style={dotted},
        minor tick num=0,
        enlarge x limits=false,
        xlabel near ticks,
        ylabel near ticks,
        xlabel style={font=\tiny},
        ylabel style={font=\tiny},
        tick label style={font=\tiny},
        legend style={
            font=\tiny,
            draw=none,
            fill=none,
            at={(0.02,0.98)},
            anchor=north west,
        },
        xmode=log,
        ymode=log,
        log basis x=10,
        xlabel={$\varepsilon$},
        ylabel={$\|S-\tilde S\|_2$},
        title={\tiny{Error versus $\varepsilon$ ($N=64$)}},
    ]
    \addplot[cyan!60!black, thick, mark=*] coordinates{
    (1.000000e-14, 6.521864e-15) (1.000000e-12, 7.052797e-13) (1.000000e-10, 8.019122e-11) (1.000000e-08, 8.152643e-09) (1.000000e-06, 8.868481e-07)
    }; \addlegendentry{Random}
    \addplot[red, thick, mark=square*] coordinates{
    (1.000000e-14, 1.427078e-15) (1.000000e-12, 2.214184e-13) (1.000000e-10, 1.513407e-11) (1.000000e-08, 1.776349e-09) (1.000000e-06, 2.080670e-07)
    }; \addlegendentry{Tridiagonal}
    \addplot[green!60!black, thick, mark=triangle*] coordinates{
    (1.000000e-14, 9.788353e-15) (1.000000e-12, 9.813416e-13) (1.000000e-10, 9.946941e-11) (1.000000e-08, 1.022867e-08) (1.000000e-06, 1.015741e-06)
    }; \addlegendentry{Exponential}
    \addplot[black, dashed, thick, mark=none] coordinates{
    (1.000000e-14, 2.828427e-14) (1.000000e-06, 2.828427e-06)
    }; \addlegendentry{$2\sqrt{2}\varepsilon$}
    \end{axis}
    \end{tikzpicture}
\caption{For $N=64$ we report how the block encoding error increases with $\varepsilon$ ranging from $1.0e-14$ to $1.0e-6$. 
        \label{fig:first}}
  \end{subfigure}
  \hfill
  \begin{subfigure}[t]{0.48\linewidth}
    \centering
    \begin{tikzpicture}
    \begin{axis}[
        width=0.95\linewidth,
        height=5.0cm,
        grid=both,
        grid style={dotted},
        minor tick num=0,
        enlarge x limits=false,
        xlabel near ticks,
        ylabel near ticks,
        xlabel style={font=\tiny},
        ylabel style={font=\tiny},
        tick label style={font=\tiny},
        legend style={
            font=\tiny,
            draw=none,
            fill=none,
            at={(0.55,0.85)},
            anchor=north east,
        },
        xmode=log,
        log basis x=2,
        xmin=4, xmax=128,
        xtick={4,8,16,32,64,128},
        xticklabels={4,8,16,32,64,128},
        xlabel={$N$},
        ylabel={$\|S-\tilde S\|_2$},
        title={\tiny{Error versus $N$ ($\varepsilon=1.0e-6$)}},
        scaled y ticks=true,
    ]
    \addplot[cyan!60!black, thick, mark=*] coordinates{
    (4,1.036912e-06) (8,1.009293e-06) (16,8.116716e-07) (32,8.462596e-07) (64,7.793250e-07) (128,6.115727e-07)
    }; \addlegendentry{Random}
    \addplot[red, thick, mark=square*] coordinates{
    (4,1.020794e-06) (8,5.495799e-07) (16,7.013596e-07) (32,3.817600e-07) (64,3.680882e-07) (128,2.406628e-07)
    }; \addlegendentry{Tridiagonal}
    \addplot[green!60!black, thick, mark=triangle*] coordinates{
    (4,1.103605e-06) (8,1.272339e-06) (16,9.178659e-07) (32,9.755029e-07) (64,9.996382e-07) (128,9.810714e-07)
    }; \addlegendentry{Exponential}
    \addplot[black, dashed, thick, mark=none] coordinates{
    (4,2.828427e-06) (128,2.828427e-06)
    }; \addlegendentry{$2\sqrt{2}\varepsilon$}
    \end{axis}
    \end{tikzpicture}
            \caption{For different values of $N$ and for $\varepsilon = 1.0e{-06}$, we report the actual errors of the encoding. We note that, for all the three different ways of selecting the generators,  the actual error is below the value $2\sqrt2 \varepsilon$ and is independent of $N$. }
        \label{fig:second}
  \end{subfigure}
  \caption{Semiseparable reconstruction error scaling: sensitivity to $\varepsilon$ (left) and $N$ (right).}
  \label{fig:err_combined}
\end{figure}

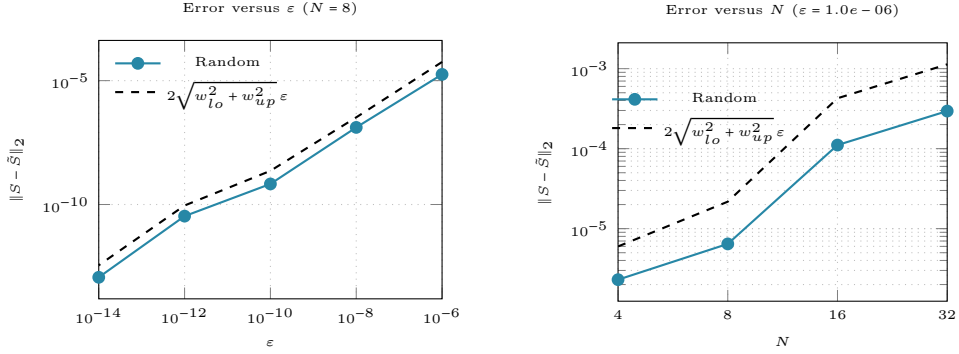
\begin{figure}[htbp]
  \centering
  \begin{subfigure}[t]{0.48\linewidth}
    \centering
    \begin{tikzpicture}
    \begin{axis}[
        width=0.98\linewidth,
        height=5.0cm,
        grid=both,
        grid style={dotted},
        minor tick num=0,
        enlarge x limits=false,
        xlabel near ticks,
        ylabel near ticks,
        xlabel style={font=\tiny},
        ylabel style={font=\tiny},
        tick label style={font=\tiny},
        legend style={
            font=\tiny,
            draw=none,
            fill=none,
            at={(0.02,0.98)},
            anchor=north west,
        },
        xmode=log,
        ymode=log,
        log basis x=10,
        xlabel={$\varepsilon$},
        ylabel={$\|S-\tilde S\|_2$},
        title={\tiny{Error versus $\varepsilon$ ($N=8$)}},
    ]
    \addplot[cyan!60!black, thick, mark=*] coordinates{
    (1.000000e-06, 1.787423e-05) (1.000000e-08, 1.295686e-07) (1.000000e-10, 6.756870e-10) (1.000000e-12, 3.393951e-11) (1.000000e-14, 1.151940e-13)
    }; \addlegendentry{Random}
    \addplot[black, dashed, thick, mark=none] coordinates{
    (1.000000e-06, 5.678374e-05) (1.000000e-08, 3.297979e-07) (1.000000e-10, 2.286000e-09) (1.000000e-12, 9.105410e-11) (1.000000e-14, 3.521653e-13)
    }; \addlegendentry{$2\sqrt{w_{lo}^2+w_{up}^2}\,\varepsilon$}
    \end{axis}
    \end{tikzpicture}
\caption{For $N=8$ we report how the block encoding error increases with $\varepsilon$ ranging from 1.0e-14 to 1.0e-06.
        \label{fig:unsym-eps}}
  \end{subfigure}
  \hfill
  \begin{subfigure}[t]{0.48\linewidth}
    \centering
    \begin{tikzpicture}
    \begin{axis}[
        width=0.95\linewidth,
        height=5.0cm,
        grid=both,
        grid style={dotted},
        minor tick num=0,
        enlarge x limits=false,
        xlabel near ticks,
        ylabel near ticks,
        xlabel style={font=\tiny},
        ylabel style={font=\tiny},
        tick label style={font=\tiny},
        legend style={
            font=\tiny,
            draw=none,
            fill=none,
            at={(0.55,0.85)},
            anchor=north east,
        },
        ymode=log,
        xmode=log,
        log basis x=2,
        xmin=4, xmax=32,
        xtick={4,8,16,32},
        xticklabels={4,8,16,32},
        xlabel={$N$},
        ylabel={$\|S-\tilde S\|_2$},
        title={\tiny{Error versus $N$ ($\varepsilon=1.0e-06$)}},
        scaled y ticks=false,
    ]
    \addplot[cyan!60!black, thick, mark=*] coordinates{
    (4.000000e+00, 2.300516e-06) (8.000000e+00, 6.442899e-06) (1.600000e+01, 1.108780e-04) (3.200000e+01, 2.950824e-04)
    }; \addlegendentry{Random}
    \addplot[black, dashed, thick, mark=none] coordinates{
    (4.000000e+00, 5.987432e-06) (8.000000e+00, 2.173163e-05) (1.600000e+01, 4.265265e-04) (3.200000e+01, 1.132699e-03)
    }; \addlegendentry{$2\sqrt{w_{lo}^2+w_{up}^2}\,\varepsilon$}
    \end{axis}
    \end{tikzpicture}
            \caption{For different values of $N$ and for $\varepsilon = 1.0e-06$, we report the actual errors of the encoding for the Random generator case, together with the theorem bound $2\sqrt{w_{lo}^2+w_{up}^2}\varepsilon$.}
        \label{fig:unsym-N}
  \end{subfigure}
  \caption{Unsymmetric semiseparable reconstruction error scaling: sensitivity to $\varepsilon$ (left) and $N$ (right).}
  \label{fig:err_combined_unsymm}
\end{figure}

Figure~\ref{fig:first} fixes $N=64$ and varies $\varepsilon$ from
$10^{-14}$ to $10^{-6}$: the error scales linearly with $\varepsilon$ and remains
below the bound throughout. The three families behave similarly, all curves lying
roughly a factor of three below the bound, which indicates that the constant
$2\sqrt{2}$ is conservative but of the correct order of magnitude. Figure~\ref{fig:second} fixes $\varepsilon=10^{-6}$ and varies $N$: the measured
error stays well below the theoretical bound $2\sqrt{2}\,\varepsilon$ for all
sizes and does not grow with $N$, consistently with the dimension-independence of
the bound.

\subsection{Comparison with FABLE}
As a general-purpose baseline we compare against FABLE~\cite{fable}, which
block-encodes an arbitrary dense matrix through a circuit of controlled rotations
whose number scales as $O(N^2)$. For our
construction we show two variants, differing only in the state-preparation
routine: a general-purpose loader following M\"ott\"onen et
al.~\cite{mottonen2005}, applicable to arbitrary generators, and a dedicated
loader for the {\sc Exponential} generators, which, being a product state,
requires only $n=\log_2 N$ single-qubit rotations for each generator.  Figure~\ref{fig:fable-a} reports depth and size of the circuits before
transpilation, where the theoretical scalings are directly visible: FABLE grows
as $2N^2$, the M\"ott\"onen loader linearly in $N$, and the exponential loader
quadratically in $\log_2 N$: indeed its depth and size are exactly
$n^2+3n+5$ and $n^2+8n+3$. All circuits are then transpiled with the Qiskit
transpiler~\cite{qiskit2024} to the gate set $\{\textsc{cnot}, X, T,
T^\dagger, S, S^\dagger, \textsc{swap}, R_y\}$, and Figure~\ref{fig:fable-b}
reports the resulting decomposed depth and total gate count on a log--log
scale. 

The three curves display two clearly different behaviours. A least-squares fit on the transpiled data 
confirms that FABLE follows the $O(N^2)$ law almost exactly (fitted exponent
$2.04$ in $N$, $R^2=0.9998$ for the depth and $2.02$ and $R^2=0.9999$ for the size), whereas both variants of our
construction grow far more slowly; the curves cross near $N=32$, beyond which
FABLE becomes more expensive. For the {\sc Exponential} loader the depth and size are well
described by a polylogarithmic model $C\,(\log_2 N)^p$ with $p\approx 2.1$
($R^2\approx 0.99$), and the corresponding growth in $N$ is strongly sublinear, confirming the polylogarithmic scaling
predicted by Theorem~\ref{thm:main}. The general-purpose M\"ott\"onen loader is more expensive than the dedicated
exponential one but still grows much more slowly than FABLE. Although its worst-case cost is linear in $N$, its transpiled depth grows
sublinearly, most likely because the Qiskit transpiler simplifies the
uniformly controlled rotations. This interpretation is supported by
Figure~\ref{fig:fable-a}, where the untranspiled count follows the expected
linear trend.
At $N=1024$ the FABLE depth and size exceed those of our exponential-loader
circuit by more than two orders of magnitude.
Since FABLE assumes no structure, the comparison quantifies the advantage of exploiting the
generator representation rather than criticizing that method.

\begin{figure}[htbp]
\centering

  \begin{subfigure}{\textwidth}
    \centering
    \begin{tikzpicture}
    \begin{groupplot}[
        group style={
            group size=2 by 1,
            horizontal sep=1.2cm,
        },
        width=0.5\linewidth,
        height=5.0cm,
        xmode=log,
        ymode=log,
        log basis x=2,
        xmin=4,
        xmax=1024,
        xtick={4,8,16,32,64,128,256,512,1024},
        xticklabels={4,8,16,32,64,128,256,512,1024},
        xlabel={$N$},
        xlabel style={font=\tiny},
        ylabel style={font=\tiny},
        title style={font=\tiny, yshift=-2pt},
        tick label style={font=\tiny},
        grid=both,
        grid style={dotted},
        minor tick num=0,
        enlarge x limits=false,
        legend style={
            font=\tiny,
            draw=none,
            fill=none,
            at={(0.02,0.98)},
            anchor=north west,
        },
    ]
    \nextgroupplot[
        ylabel={Circuit depth},
        title={Undecomposed circuit depth},
    ]
    \addplot[magenta, thick, mark=diamond*]coordinates{ 
    (4,23) (8,47) (16,89) (32,165) (64,307) (128,579) (256,1109) (512,2153) (1024,4223)
    }; \addlegendentry{Möttönen}
    \addplot[green!60!black, thick, mark=triangle*] coordinates{
    (4,15) (8,23) (16,33) (32,45) (64,59) (128,75) (256,93) (512,113) (1024,135)
    }; \addlegendentry{Exponential}
     \addplot[blue, thin, mark=o] coordinates{
    (4,26) (8,98) (16,456) (32,1922) (64,7879) (128,31894) (256,127914) (512,517089) (1024,2028835)
    }; \addlegendentry{FABLE}
    \addplot[black, densely dotted, thick] coordinates{
    (4,26) (1024,1703936)
    }; \addlegendentry{$\propto N^2$}
    
    \nextgroupplot[
        ylabel={Circuit size (gate count)},
        title={Undecomposed circuit size},
    ]
    \addplot[magenta, thick, mark=diamond*] coordinates{  
    (4,39) (8,82) (16,159) (32,302) (64,575) (128,1106) (256,2151) (512,4222) (1024,8343)
    }; \addlegendentry{Möttönen}
     \addplot[green!60!black, thick, mark=triangle*] coordinates{
     (4,23) (8,36) (16,51) (32,68) (64,87) (128,108) (256,131) (512,156) (1024,183)
    }; \addlegendentry{Exponential}
     \addplot[blue, thin, mark=o] coordinates{
       (4,30) (8,105) (16,466) (32,1935) (64,7895) (128,31913) (256,127936) (512,517114) (1024,2028863)
    }; \addlegendentry{FABLE}
    \addplot[black, densely dotted, thick] coordinates{
    (4,30) (1024,1966080)
    }; \addlegendentry{$\propto N^2$}
    \end{groupplot}
    \end{tikzpicture}
    \caption{Pre-transpilation depth and size count.} \label{fig:fable-a}
  \end{subfigure}


  \begin{subfigure}{\textwidth}
    \centering
    \begin{tikzpicture}
    \begin{groupplot}[
        group style={
            group size=2 by 1,
            horizontal sep=1.2cm,
        },
        width=0.5\linewidth,
        height=5.0cm,
        xmode=log,
        ymode=log,
        log basis x=2,
        xmin=4,
        xmax=1024,
        xtick={4,8,16,32,64,128,256,512,1024},
        xticklabels={4,8,16,32,64,128,256,512,1024},
        xlabel={$N$},
        xlabel style={font=\tiny},
        ylabel style={font=\tiny},
        title style={font=\tiny, yshift=-2pt},
        tick label style={font=\tiny},
        grid=both,
        grid style={dotted},
        minor tick num=0,
        enlarge x limits=false,
        legend style={
            font=\tiny,
            draw=none,
            fill=none,
            at={(0.02,0.98)},
            anchor=north west,
        },
    ]
    \nextgroupplot[
        ylabel={Circuit depth},
        title={Decomposed circuit depth},
    ]
    \addplot[magenta, thick, mark=diamond*]coordinates{ 
    (4,294) (8,661) (16,1204) (32,2005) (64,3163) (128,4930) (256,7627) (512,11937) (1024,19207)
    }; \addlegendentry{Möttönen}
     \addplot[green!60!black, thick, mark=triangle*] coordinates{
    (4,268) (8,595) (16,1058) (32,1699) (64,2537) (128,3664) (256,5081) (512,6831) (1024,8981)
    }; \addlegendentry{Exponential}
    \addplot[blue, thin, mark=o] coordinates{
    (4,26) (8,98) (16,456) (32,1922) (64,7879) (128,31894) (256,127914) (512,517089) (1024,2028835)
    }; \addlegendentry{FABLE}
    \addplot[black, densely dotted, thick] coordinates{
    (4,26) (1024,1703936)
    }; \addlegendentry{$\propto N^2$}
    
    \nextgroupplot[
        ylabel={Circuit size (gate count)},
        title={Decomposed circuit size},
    ]
    \addplot[magenta, thick, mark=diamond*]coordinates{ 
    (4,417) (8,944) (16,1737) (32,2918) (64,4689) (128,7412) (256,11769) (512,19082) (1024,31953)
    }; \addlegendentry{Möttönen}
      \addplot[green!60!black, thick, mark=triangle*] coordinates{
     (4,365) (8,814) (16,1449) (32,2312) (64,3445) (128,4890) (256,6689) (512,8884) (1024,11517)
    }; \addlegendentry{Exponential}
        \addplot[blue, thin, mark=o] coordinates{
    (4,30) (8,105) (16,466) (32,1935) (64,7895) (128,31913) (256,127936) (512,517114) (1024,2028863)
    }; \addlegendentry{FABLE}
    \addplot[black, densely dotted, thick] coordinates{
    (4,30) (1024,1966080)
    }; \addlegendentry{$\propto N^2$}
    \end{groupplot}
    \end{tikzpicture}
    \caption{Post-transpilation depth and size count.} \label{fig:fable-b}
  \end{subfigure}
\caption{Comparison of our block encoding with FABLE~\cite{fable} in terms of
circuit depth (left) and size (right), for semiseparable matrices with
{\sc Exponential} generators, $u_i=e^{\theta t_i}$, $v_i=e^{-\theta t_i}$. All
panels show two variants of our construction, differing only in the
state-preparation routine used to load the generators: the general-purpose
scheme of M\"ott\"onen et al.~\cite{mottonen2005}, of worst-case linear cost, and a
dedicated loader (labelled {\sc Exponential}) that exploits the product
structure of the generators to prepare them with a circuit of depth $2n$.
FABLE and the M\"ott\"onen variant grow polynomially in $N$ (with FABLE
following the $O(N^2)$ reference line), whereas the dedicated exponential loader
grows polylogarithmically; for the transpiled circuits, the two regimes cross near $N=32$, beyond which our
construction is markedly cheaper.} 
  \label{fig:fable}
\end{figure}
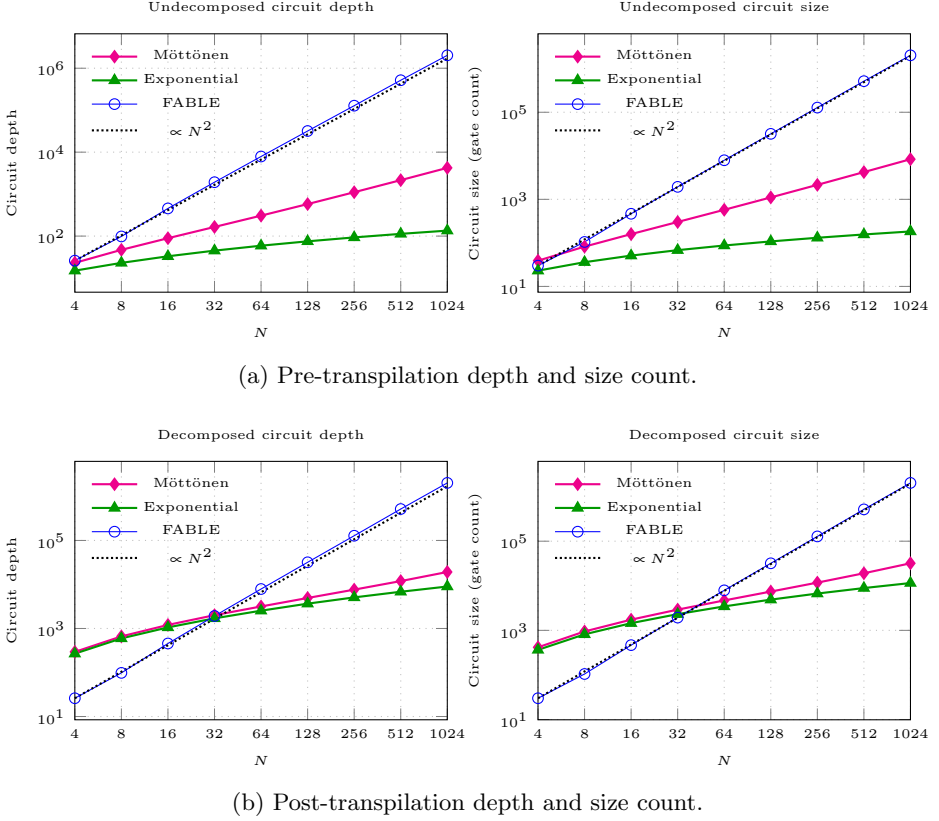

\section{Conclusions}\label{sec:concl}
We have presented an explicit quantum circuit that block-encodes a
generator-representable semiseparable matrix directly from its two generators,
without forming any factorization or using complicated hierarchical decompositions. For an $N\times N$ one-pair matrix with
$N=2^{n}$, the construction uses $2n+2$ ancilla qubits, has $O(\polylog N)$
depth, and realizes a $(2\sqrt{N},\,2n+2,\,2\sqrt{2}\,\varepsilon)$-block
encoding, where $\varepsilon$ is the accuracy of the state preparation that loads
the generators.
The same construction extends to the unsymmetric one-pair class, at the cost of a
controlled state preparation and of replacing the Hadamard gates on the triangular
qubit by a suitable $R_y$ rotation. In general, our results illustrate that the
data-sparsity of a matrix can be turned into a concrete quantum advantage, mirroring the role that structure plays in classical fast algorithms.

Several extensions appear within reach of the techniques developed here. The extension to $r$-semiseparable matrices, where the rank of upper and lower part is at most $r$, to generator representable plus diagonal, a subclass of quasiseparable matrices~\cite{eidelman2014separable}, is straightforward combining our construction through a two-term linear combination of unitaries. However, the extension to non-generator-representable structure is still an open problem.

The one-pair case is the flat,
single-level instance of the broader family of rank-structured matrices. It would be  interesting to explore whether the idea behind our algorithm can be applied recursively
across the levels of an HODLR or HSS partition, so as to build a block encoding
of a hierarchically semiseparable matrix without resorting to a level-by-level
linear combination of unitaries~\cite{NKL22,TL26}.

\section*{Declarations}
During the preparation of this work, the authors used Claude (Anthropic) to improve the readability of the text and to assist in identifying relevant literature. After using this tool, the authors reviewed and edited the content as needed and take full responsibility for the content of the publication.

\subsubsection*{Fundings}
G. Antonioli and G. Del Corso  carried out
this study  within the National Centre on HPC, Big Data and Quantum Computing - SPOKE 10 (Quantum Computing) and received funding from the European Union Next-GenerationEU - National Recovery and Resilience Plan (NRRP) – MISSION 4 COMPONENT 2, INVESTMENT N. 1.4 – CUP N. I53C22000690001. 

Partial support for these authors was also provided by the INdAM - GNCS project ``Strutture di matrici e di funzioni per la sintesi di circuiti quantistici efficienti'', CUP E53C23001670001.

G. Del Corso is also partially supported by European Union - NextGenerationEU under the National Recovery and Resilience Plan (PNRR) - Mission 4 Education and research - Component 2 From research to business - Investment 1.1 Notice Prin 2022 - DD N. 104 2/2/2022, titled Low-rank Structures and Numerical Methods in Matrix and Tensor Computations and their Application, proposal code 20227PCCKZ – CUP I53D23002280006.
J53D23003620006. 
G. Del Corso is also partially supported by  the U.S. Department of Energy, Office of Science, National Quantum Information Science Research Centers, Superconducting Quantum Materials and Systems Center (SQMS), under Contract No. 89243024CSC000002.

\subsubsection*{Ethical Approval}
It complies with the Ethical Rules applicable for this journal.
\subsubsection*{Data Availability}
No datasets were generated or analyzed during the current study.
\subsubsection*{Code Availability} 
The code of the implemented algorithm is available at\\ \url{https://github.com/Giacomo-Antonioli/Semiseparable_Block_Encoding}.
\subsubsection*{Competing interests}
The authors declare no competing interests.
\bibliographystyle{plain}
\bibliography{biblio}

\end{document}